\documentclass[final, 10pt, letterpaper]{IEEEtran}
\usepackage[T1]{fontenc}
\usepackage{url}
\ifCLASSINFOpdf
\else
\fi

\usepackage{amsmath}
\usepackage{algorithmicx}
\usepackage{algorithm}
\usepackage{graphicx}
\usepackage{mathtools}
\usepackage{amsthm}
\usepackage{amssymb}
\usepackage{algpseudocode}
\usepackage{nccmath,textcomp}
\usepackage{eso-pic}
\usepackage{bbm}
\usepackage{xcolor}
\newtheorem{theorem}{Theorem}
\newtheorem{lemma}{Lemma}
\newtheorem{remark}{Remark}
\newtheorem{comment}{Comment}

\newtheorem{proposition}{Proposition}
\newtheorem{assumption}{Assumption}
\newtheorem{corollary}{Corollary}
\newtheorem*{proof-mine}{Proof}
\usepackage{verbatim}

\makeatletter
\let\MYcaption\@makecaption
\makeatother

\usepackage[font=footnotesize]{subcaption}

\makeatletter
\let\@makecaption\MYcaption
\makeatother

\begin{document}
\title{Underlay Radar-Massive MIMO Spectrum Sharing: Modeling Fundamentals and Performance Analysis}
\author{Raghunandan M. Rao,~\IEEEmembership{Member,~IEEE,} Harpreet S. Dhillon,~\IEEEmembership{Senior Member,~IEEE,} Vuk Marojevic,~\IEEEmembership{Senior Member,~IEEE,} Jeffrey H. Reed,~\IEEEmembership{Fellow,~IEEE} \\ 
\thanks{R. M. Rao, H. S. Dhillon and J. H. Reed are with Wireless@VT, Bradley Department of ECE, Virginia Tech, Blacksburg, VA, 24061, USA (e-mail: \{raghumr, hdhillon, reedjh\}@vt.edu). V. Marojevic is with the Department of ECE at Mississippi State University, Mississippi State, MS, 39762, USA (e-mail: vuk.marojevic@ece.msstate.edu). The support of the U.S. NSF Grants CNS-1564148, CNS-1642873, and ECCS-1731711 is gratefully acknowledged. This work was presented in part at IEEE Globecom, Waikoloa, HI, USA, 2019 \cite{Rao_Dhillon_Globecom_2019}.}
}
\maketitle
\vspace{-45pt}
\thispagestyle{empty}
\begin{abstract}
Spectrum sharing alleviates the severe shortage of spectrum in sub-6 GHz frequency bands through the harmonious coexistence of two or more wireless technologies on the same frequency resources. In this work, we study underlay radar-massive MIMO cellular coexistence in LoS/near-LoS channels, where both systems have 3D beamforming capabilities. Using mathematical tools from stochastic geometry, we derive an upper bound on the average interference power at the radar due to the 3D massive MIMO cellular downlink under the worst-case `cell-edge beamforming' conditions. To overcome the technical challenges imposed by asymmetric and arbitrarily large cells, we devise a novel construction in which each Poisson Voronoi (PV) cell is bounded by its circumcircle to bound the effect of the random cell shapes on average interference. Since this model is intractable for further analysis due to the correlation between adjacent PV cells' shapes and sizes, we propose a tractable nominal interference model, where we model each PV cell as a circular disk with an area equal to the average area of the typical cell. We quantify the gap in the average interference power between these two models and show that the upper bound is tight for realistic deployment parameters. We also compare them with a more practical but intractable MU-MIMO scheduling model to show that our worst-case interference models show the same trends and do not deviate significantly from realistic scheduler models. Under the nominal interference model, we characterize the interference distribution using the dominant interferer approximation by deriving the equi-interference contour expression when the typical receiver uses 3D beamforming. Finally, we use tractable expressions for the interference distribution to characterize radar's spatial probability of false alarm/detection in a quasi-static target tracking scenario. Our results reveal useful trends in the average interference as a function of the deployment parameters (BS density, exclusion zone radius, antenna height, transmit power of each BS, etc.). We also provide useful system design insights using radar receiver operating characteristic (ROC) curves by applying our analytical results to design the minimum exclusion zone radius in current and future radar-cellular spectrum sharing scenarios.
\end{abstract}
\begin{IEEEkeywords}
Stochastic geometry, radar-massive MIMO coexistence, interference distribution, probability of false alarm, probability of detection.
\end{IEEEkeywords}

\IEEEpeerreviewmaketitle
\vspace{- 5pt}
\section{Introduction}
Due to exponential growth in traffic over the last decade, mobile network operators (MNOs) in the United States have paid a total of more than $\$$60 billion to obtain licensed spectrum in the heavily congested sub-6 GHz bands \cite{FCC_AWS3_auction}. Spectrum sharing in 5G networks has the potential to alleviate this shortage, where the goal is to enable the operation of multiple wireless technologies on the same frequency bands without negatively impacting each others' performance. Spectrum sharing in sub-6 GHz bands is especially effective since it is underutilized. As a result, spectrum sharing policies in these bands have been ratified by the Federal Communications Commission in the United States \cite{FCC_3point5_GHz_Rules}, \cite{FCC_5_GHz_FirstOrder}. 
	
Radar systems are the most prominent incumbents in sub-6 GHz frequencies \cite{Griffiths_Radar_tech_regul_PRoc_IEEE_2015}, and their sparsity in space/time provides excellent opportunities for spectrum sharing.  For example, weather, air-traffic control (ATC), and military radars use a pulsed waveform and scan for targets using a \textit{rotating pattern} \cite{Khan_DaSilva_WCM_2016}, \cite{Hessar_Roy_Radar_WiFi_TAES_2016}. Opportunistic spectrum sharing schemes leverage the periodic rotation to schedule transmissions during interference-free durations. Similarly, cellular systems can operate in pulsed interference without significant performance degradation by using the interference-free durations between pulses \cite{Rao_Vuk_TVT_2020}. On the other hand, the radar can operate without significant performance degradation as long as the interference-to-noise-ratio (INR) is below the threshold \cite{sanders2006effects}. Also, radar systems are deployed in a spatially sparse manner, when compared to other systems such as Wi-Fi \cite{Li_Baccelli_Andrews_LTE_WiFi_TWC_2016}, \cite{Hessar_Roy_Radar_WiFi_TAES_2016}, Heterogeneous Networks (HetNets) \cite{Bodong_SSS_HetNet_TWC_2020}, and device-to-device (D2D) communication systems \cite{Chen_Dhillon_Liu_QoS_D2D_StochGeom_TCOM_2019}, \cite{Bodong_Liu_UAV_D2D_TCOM_2020}. 

In the wireless industry, these factors have encouraged commercial standardization and deployment of radar-cellular spectrum sharing systems in the 3.5 GHz band, led by the Citizens Broadband Radio Service (CBRS) Alliance \cite{CBRS_Alliance}. Its efforts will lead to the deployment of spectrum sharing solutions between LTE/5G and naval radar systems in the 3.5 GHz band on the coastal areas of the continental United States, which is projected to incur \$1 billion in infrastructure investment annually by 2023 \cite{SNS_Telecom_CBRS_Projection}. Due to the same reasons, radar-cellular spectrum sharing has also caught the interest of the defense community. The United States Department of Defense has undertaken the first step to incorporate 5G into their mission-critical networks. By allocating hundreds of millions of dollars to its partners in the wireless industry, it intends to build large-scale experimental testbeds to enable its terrestrial and airborne radar platforms to share spectrum with 5G systems in the 3.1-3.45 GHz band \cite{US_DoD_Spec_Share_budget}.

Among the different approaches, underlay radar-cellular spectrum sharing is an important baseline, wherein a large exclusion zone is established around the radar. Only the cellular base stations outside the exclusion zone are allowed to operate, often \textit{without coordination}. Spectrum sharing is feasible in these scenarios as long as the aggregate interference lies below a threshold. This paper develops a mathematical framework using stochastic geometry to model underlay radar-massive MIMO spectrum sharing scenarios and study the impact of the worst-case cellular interference on the radar's detection and false alarm performance.
\vspace{-10pt}
\subsection{Related Work}
Prior works have considered different approaches for radar-cellular coexistence, such as multi-antenna techniques, waveform design, and opportunistic spectrum sharing. Multi-antenna techniques use the additional spatial degrees of freedom to mitigate mutual interference between the radar and cellular system \cite{Liu_Robust_MIMO_BF_Rad_Cell_Coexist_2017}, \cite{Biswas_FDMIMO_radar_coexist_TWC_2018}. On the other hand, the waveform of the radar \cite{Bica_Mitra_ICASSP_MI_Rad_Wfrm_2016}, \cite{Tang_Li_Spectr_constr_Rad_Wfm_TSP_2019} and cellular system \cite{Carrick_Reed_FRESH_TAES_2019} can be designed to enhance the receiver's resilience to interference. Finally, opportunistic spectrum sharing schemes improve the secondary system (cellular) performance by exploiting information of the temporal/spectral/spatial variation of primary user interference  \cite{Hessar_Roy_Radar_WiFi_TAES_2016}, \cite{Rao_Vuk_TVT_2020}. Accurate channel state information (CSI) is crucial for multi-antenna techniques to be effective, for which cooperative schemes such as common knowledge of radar and cellular probing waveforms is necessary \cite{Liu_Geraci_ICSI_HowMany_TWC_2019}. However, security concerns make cooperation infeasible with some military radar systems. Meanwhile, the adoption of interference-resilient waveforms has been very slow, since they require significant modifications to both systems, making their mass deployment over the next few years unlikely. While opportunistic spectrum access is feasible in the case of rotating radars in the `search mode' \cite{Hessar_Roy_Radar_WiFi_TAES_2016}, it is not possible when the radar is tracking a target. Therefore, an essential baseline of radar-cellular spectrum sharing is the underlay mode; a static exclusion zone is defined around the radar to limit cellular interference below a predefined threshold in the absence of cooperation.

In practice, protocol-oriented system-level simulators are used to undertake feasibility studies before proceeding with testbed-based experimentation and deployment. However, the use of 3D beamforming-capable massive MIMO in 5G base stations \cite{Xu_FDMIMO_Samsung_JSAC_2017}, and the presence of large exclusion zones \cite{Sudeep_Abid_DynExcZone_DySPAN_2014}, significantly increase the computational complexity of the system-level simulator, resulting in lengthy execution times. Over the last few years, stochastic geometry has augmented simulation studies by providing a tractable mathematical framework to gain fundamental system-design insights.  Due to its analytical tractability, stochastic geometry has become a useful tool to analyze large scale behavior of spectrum sharing scenarios such as LTE-WiFi coexistence \cite{Li_Baccelli_Andrews_LTE_WiFi_TWC_2016}, \cite{Parida_Dhillon_CBRS_Access_2017}, radar-WiFi coexistence \cite{Hessar_Roy_Radar_WiFi_TAES_2016}, HetNets \cite{Bodong_SSS_HetNet_TWC_2020}, and cellular-D2D coexistence \cite{Chen_Dhillon_Liu_QoS_D2D_StochGeom_TCOM_2019}, \cite{Bodong_Liu_UAV_D2D_TCOM_2020}.

Typically, radar systems are sparsely deployed with a large exclusion zone established around it when sharing spectrum with other technologies. This is significantly different when compared to spectrum sharing scenarios in \cite{Li_Baccelli_Andrews_LTE_WiFi_TWC_2016}, \cite{Chen_Dhillon_Liu_QoS_D2D_StochGeom_TCOM_2019}, \cite{Bodong_Liu_UAV_D2D_TCOM_2020}, \cite{Parida_Dhillon_CBRS_Access_2017} because the density of coexisting systems tend to be much higher. Furthermore, these systems tend to be closely located to one another and use spectrum sensing techniques (overlay spectrum sharing) instead of large exclusion zones to limit interference. The channel conditions are also significantly different; multipath is significant in dense cellular-WiFi/D2D deployments, whereas some radar-cellular coexistence scenarios are characterized by LoS \cite{Liu_Geraci_ICSI_HowMany_TWC_2019} or Rician channels \cite{Biswas_FDMIMO_radar_coexist_TWC_2018}, especially in coastal deployments and also when radar systems use ground clutter suppression techniques \cite{Melvin_STAP_TAESM_2004}. In radar-cellular coexistence where both systems are equipped with 3D beamforming capabilities \cite{LTEAPro_FDMIMO_Samsung_ComMag_2017}, modeling the impact of azimuth as well as elevation beamforming gains are crucial to accurately model the received interference power. However, most of the prior work in stochastic geometry consider uniform linear arrays with \textit{only azimuth beamforming capabilities}, and the beamforming pattern is approximated by a piecewise constant function, often obtained from the main lobe and the two side lobe gains \cite{Bai_Heath_mmWave_TWC_2015}, or the exact beamforming pattern \cite{Hessar_Roy_Radar_WiFi_TAES_2016}, \cite{Kim_WiFi_Radar_WCL_2017}. Even though some recent works in stochastic geometry account for the elevation beamforming gain in their analysis, the models aren't well-suited for analytical treatment \cite{Rebato_Park_Zorzi_AntPattrn_mmWave_TCOM_2019}, \cite{Kim_Visotsky_5G_mmWave_Coexist_JSAC_2017}, or focus on fixed downtilt scenarios for optimal coverage in multi-cellular networks  \cite{Yang_BS_Downtilt_DL_Cellular_TWC_2019}, which is different from radar-cellular coexistence due to the aforementioned reasons. 
\vspace{-10pt}
\subsection{Contributions}
In this work, we develop a novel and tractable analytical framework to analyze radar performance metrics in a radar-massive MIMO spectrum sharing scenario. We consider a single radar system located at the origin, tracking a target above the horizon using a single beam from a uniform rectangular array (URA). The radar is surrounded by massive MIMO BSs, which are distributed as a homogeneous Poisson point process (PPP). All BSs are equipped with a massive MIMO URA mounted at the same height, where each BS is serving multiple users in its cell using hybrid 3D beamforming \cite{JSDM_Adhikary_Caire_TIT_2013}. A circular exclusion zone (EZ) is established around the radar, and only the BSs lying outside the EZ are allowed to operate.
\subsubsection*{Worst-Case Average Interference Power}
The main objective is to model the worst-case interference at the radar due to 3D beamforming in each cell. Worst-case interference occurs when BSs serve edge users located in the general direction of the radar. But in a random network of BSs, the notion of the `cell-edge' is unclear, because Poisson-Voronoi (PV) cells are radially asymmetric and can be arbitrarily large. State-of-the-art works focusing on cellular network performance overcome this challenge by analyzing the performance at a typical user \cite{Rebato_Park_Zorzi_AntPattrn_mmWave_TCOM_2019}, \cite{Yang_BS_Downtilt_DL_Cellular_TWC_2019}. Due to our differing objectives when compared to prior works, we devise a novel formulation by bounding the random effects of asymmetric and irregular cell shapes, termed as the \textit{Circumcircle-based cell (CBC) model}. In addition, the presence of sidelobes result in a beamforming gain that is a non-monotonic function of the elevation angle. We derive an upper bound on the beamforming gain that monotonically decreases with the elevation angle, which is crucial to deriving the upper bound on the worst-case average interference. To develop a tractable and easy-to-use approximation, we also derive the nominal average interference power by modeling each PV cell as a circle of area equal to the average area of a typical cell, termed as the \textit{Average Area-Equivalent Circular Cell (AAECC) model}. Finally, we provide approximations, that lead to the development of new system design insights regarding the worst-case exclusion zone radius, scaling laws, and the gap between the worst-case and nominal average interference.

\subsubsection*{Interference Distribution}
Since we are interested in understanding the worst-case radar performance in near-LoS channel conditions, we need a different approach compared to those presented in \cite{Rebato_Park_Zorzi_AntPattrn_mmWave_TCOM_2019}, \cite{Yang_BS_Downtilt_DL_Cellular_TWC_2019}. In this paper, we use the dominant interferer approach \cite{Parida_Dhillon_CBRS_Access_2017}, \cite{Vishnu_Dhillon_UAVs_TCOM_2017} due to its tractability, and the property that it can be used to upper bound the interference power's CDF \cite{Vishnu_Dhillon_UAVs_TCOM_2017}. Among the two cell-shape models, the CBC model is intractable since it induces correlation in the circumradii of adjacent PV cells. Therefore, we use the AAECC model to derive an approximate but accurate expression for the interference distribution. Even then, this approach is non-trivial since receive beamforming at the radar URA distorts the radial symmetry of the equi-interference contour, unlike the case of omnidirectional reception where it is a circle \cite{Heath_Kontouris_Bai_Dom_Int_TSP_2013}. A novel intermediate result is the derivation of the \textit{equal interference contour}, which resembles a 2D slice of the 3D radar beamforming pattern, when the exclusion zone radius is much larger than the BS antenna deployment height. We use this to characterize the total interference distribution in terms of that of the \textit{farthest contour distance from the radar}, and provide insights regarding the accuracy of this method. 
\subsubsection*{Radar Performance Metrics}
Under a Gaussian signaling scheme, we characterize the radar detection and false alarm probabilities averaged over the BS point process \cite{Chen_Dhillon_Liu_QoS_D2D_StochGeom_TCOM_2019} in a quasi-static target scenario. We derive the exact probabilities, and develop accurate approximations using the dominant interferer method and the central limit theorem. 
Performance trends, and system design insights are discussed using radar receiver operating characteristic (ROC) curves, and we demonstrate a simple application of our analytical expressions to design the minimum exclusion zone radius in a radar-massive MIMO spectrum sharing scenario. 
\vspace{-10pt}
\section{System Model} \label{Sec_Sys_Model}
We consider the radar-massive MIMO spectrum sharing scenario shown in Fig. \ref{Fig1_Radar_mMIMO_SpecShare_Illustration}. The radar is the primary user (PU), equipped with a $N^{(\mathtt{rad})}_\mathtt{az} \times N^{(\mathtt{rad})}_\mathtt{el}$ uniform rectangular array (URA) with $\tfrac{\lambda}{2}$-spacing, mounted at a height of $h_\mathtt{rad}$ m. The massive MIMO downlink is the secondary user (SU), with each BS serving $K$ users with equal power allocation using multi-user MIMO (MU-MIMO). Each BS is equipped with a $N^{(\mathtt{BS})}_\mathtt{az} \times N^{(\mathtt{BS})}_\mathtt{el}$ URA with $\tfrac{\lambda}{2}$-spacing, mounted at a height of $h_\mathtt{BS}$ m. The subscripts $\mathtt{az}$ ($\mathtt{el}$) are used to denote the azimuth (elevation) elements respectively, and superscripts $\mathtt{rad}$ ($\mathtt{BS}$) denote the radar (BS) antenna elements respectively. The radar is assumed to be located at the origin, and protected from SU interference by a \textit{circular exclusion zone} of radius $r_\mathtt{exc}$. The exclusion zone is chosen to be circular since there is no coordination between the cellular network and the radar system, and the radar is assumed to search for a target in the azimuth $[-\tfrac{\pi}{2}, \tfrac{\pi}{2})$, as shown in Fig. \ref{Fig_Radar_mMIMO_SpecSharing_Illustrate_Assumptions}. The spatial distribution of the massive MIMO BSs is modeled as a homogeneous PPP $\mathbf{\Phi_{BS}}$, of intensity $\lambda_\mathtt{BS}$. The set of locations in the exclusion zone of radius $r_\mathtt{exc}$ is denoted by the set $\mathcal{A}_\mathtt{exc} = \{(x,y)|(x^2 + y^2) \leq r^2_\mathtt{exc}\} \subset \mathbb{R}^2$ and hence, the BS locations outside the exclusion zone in the azimuth $[-\tfrac{\pi}{2}, \tfrac{\pi}{2})$ (henceforth termed as the `interference region') is denoted by the set $\mathbf{\Phi_{int}} = \{(x,y)| x^2 + y^2 \geq r^2_\mathtt{exc}, -\tfrac{\pi}{2} \leq \theta \leq \tfrac{\pi}{2}\}$. In this work, we are interested in characterizing the \textit{worst-case interference} to the radar due to the cellular network. The users' spatial distribution is not necessary to characterize the worst-case interference, because the highest possible interference power is transmitted to the radar when each BS schedules the cell-edge user in the cell. We term this as the `cell-edge beamforming model,' where we implicitly assume that a cell-edge user is always served in each cell.

\begin{figure*}[t]
	\centering
	\begin{subfigure}[t]{0.47\textwidth}
		\raggedleft
		\includegraphics[width=2.8in]{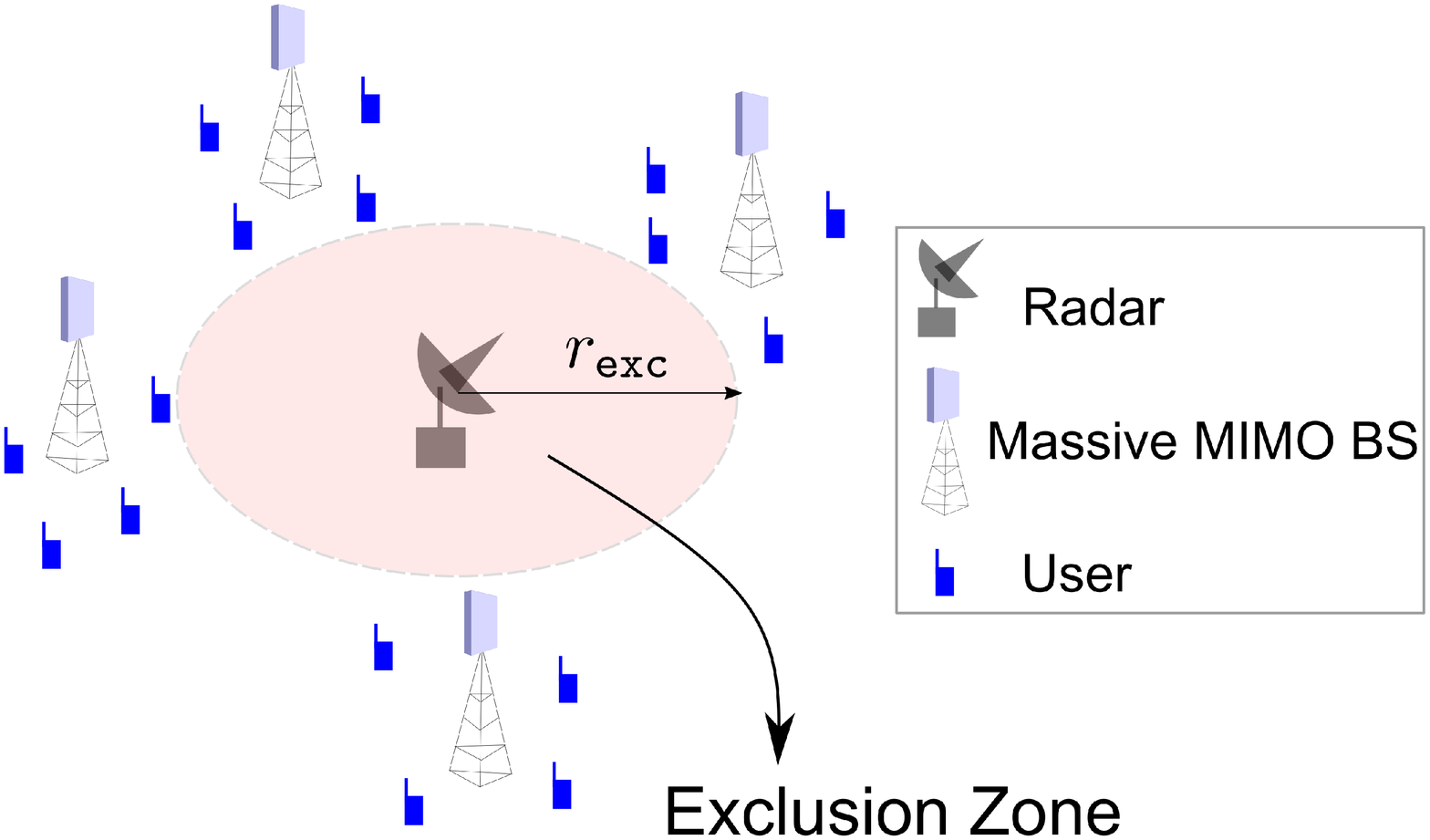}
		\caption{}
		\label{Fig1_Radar_mMIMO_SpecShare_Illustration}
	\end{subfigure}
	~
	\begin{subfigure}[t]{0.48\textwidth}
		\centering
		\includegraphics[width=3.2in]{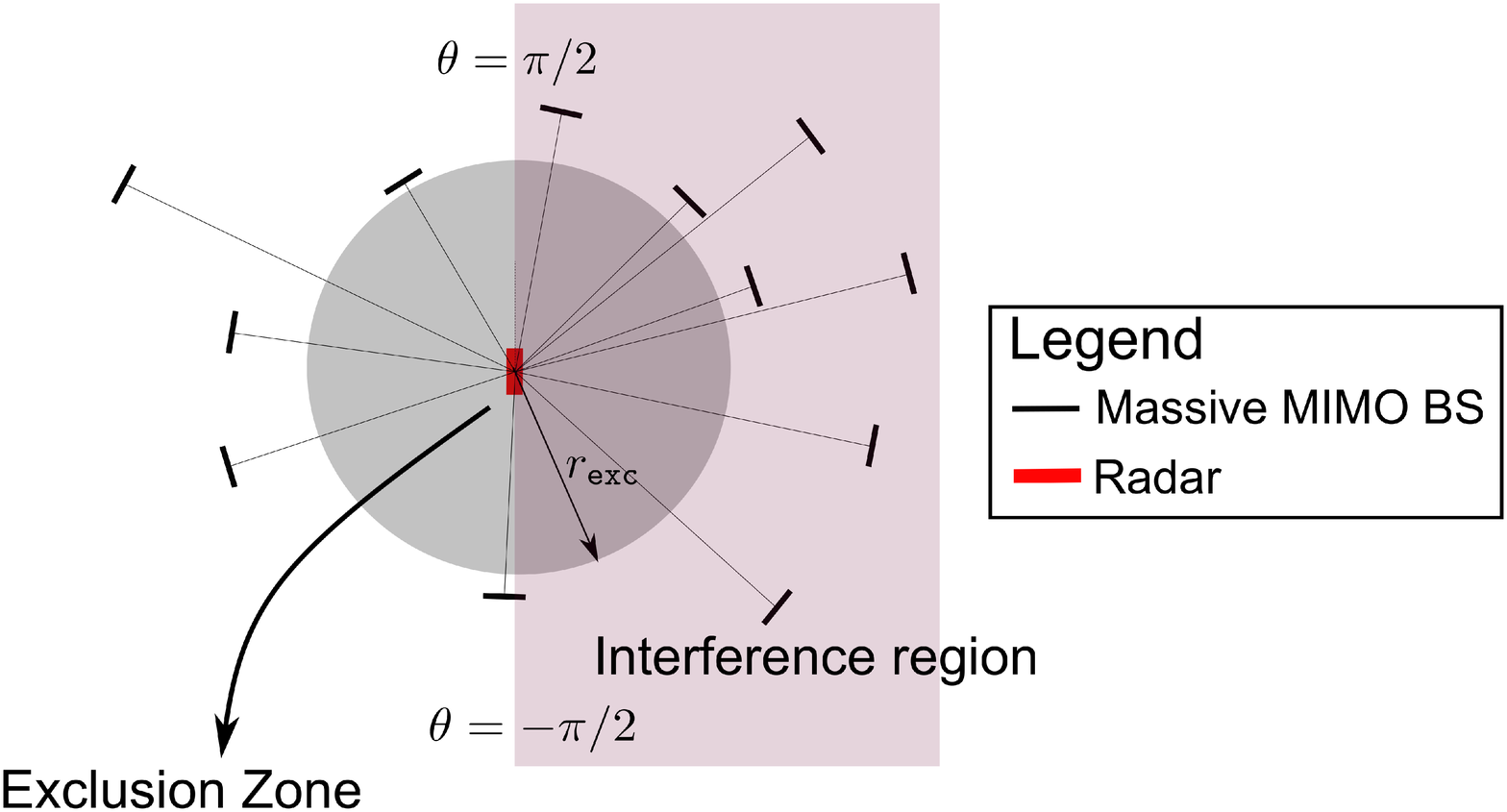}
		\caption{}
		\label{Fig1_Radar_mMIMO_SpecShare_TopView}		
	\end{subfigure}
	~
	\begin{subfigure}[t]{0.9\textwidth}
		\centering
		\includegraphics[width=4.0in]{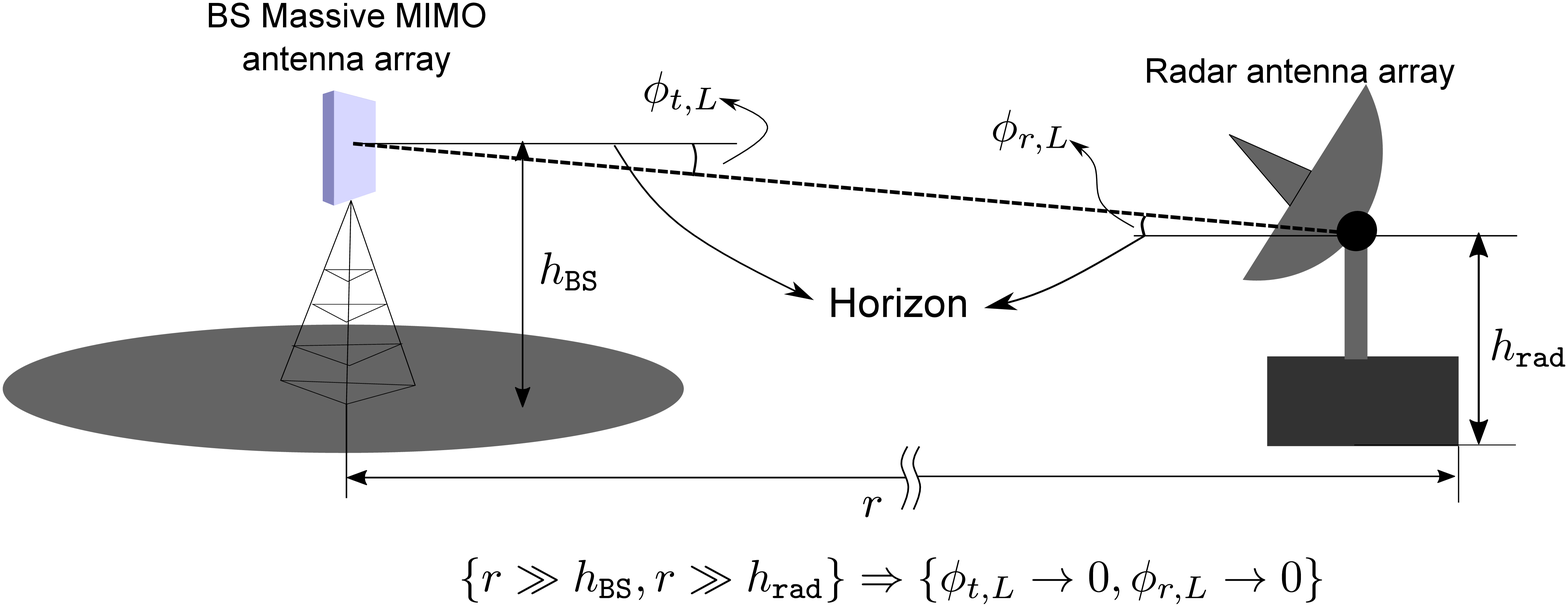}
		\caption{}
		\label{Fig1_Radar_mMIMO_SpecShare_Assumptions}		
	\end{subfigure}
	\caption{(a) Illustration of the radar-massive MIMO spectrum sharing scenario. The radar is protected from massive MIMO downlink interference by an exclusion zone of radius $r_\mathtt{exc}$. (b) Top View: the boresight of each BS is aligned along the direction of the radar, and the radar receives interference from the azimuth $[-\pi/2, \pi/2]$ depicted by the shaded region. (c) The LoS component has elevation angle of departure ($\phi_{t,L}$) and arrival ($\phi_{r,L}$) close to $0^\circ$, i.e. the horizon. In our convention, $-\pi/2 \leq \phi < 0^\circ$ for elevation angles above the horizon, and $0 < \phi \leq \pi/2$ for elevation angles below the horizon. }
	\label{Fig_Radar_mMIMO_SpecSharing_Illustrate_Assumptions}
\end{figure*}

\vspace{-10pt}
\subsection{Channel Model}
In quasi-stationary conditions, the channel between each BS and the radar is given by \cite{3GPP5GNR_ChanModels}
\begin{align}
	\label{Doub_dir_chan_model_f}
	\mathbf{H_R} = & \sqrt{\tfrac{\beta(d)}{1 + K_R}} \Big( \sqrt{K_R} \mathbf{a} (\theta_{t,L}, \phi_{t,L}) \mathbf{a}^H (\theta_{r,L}, \phi_{r,L}) + \nonumber \\
	&  \sqrt{\tfrac{1}{N_c}} \sum\limits_{i=1}^{N_c} \gamma_i \mathbf{a} (\theta_{t,i}, \phi_{t,i}) \mathbf{a}^H (\theta_{r,i}, \phi_{r, i}) \Big),  
\end{align}
where $\beta(d)= PL(r_0) d^{-\alpha}$ is the path loss at distance $d$, $PL(r_0)$ is the path-loss at reference distance $r_0$, $\alpha$ is the path-loss exponent ($\alpha > 2$), $d$ is the 3D distance between the BS and the radar, and $N_c$ is the number of discrete multipath components (MPCs). The Rician factor $K_R\gg 1$, where propagation is dominated by the LoS component\footnote{Such propagation scenarios are observed in (a) coastal deployments (for e.g., BSs sharing spectrum with a naval radar), (b) terrestrial deployments in flat rural/suburban terrain (for e.g., terrestrial BSs sharing spectrum with a terrestrial radar), and (c) radar systems that use ground clutter-suppressing techniques \cite{Melvin_STAP_TAESM_2004}.}. In addition, the random small-scale fading amplitude satisfies $\mathbb{E}[\gamma_i] = 0$ and $\mathbb{E}[|\gamma_i|^2] = 1$. The azimuth and elevation angles of arrival (departure) of the $i^{th}$ MPC at the radar (from the BS) is denoted by $\theta_{r,i}$ ($\theta_{t,i}$) and $\phi_{r,i}$ ($\phi_{t,i}$), respectively. Similarly, the azimuth and elevation angles of departure (arrival) of the LoS component are given by $\theta_{t,L}$ ($\theta_{r,L}$) and $\phi_{t,L}$ ($\phi_{r,L}$), respectively, as shown in Fig. \ref{Fig1_Radar_mMIMO_SpecShare_Assumptions}. The steering vector $\mathbf{a}(\theta_t, \phi_t) \in \mathbb{C}^{M_\mathtt{BS} \times 1}$ (BS), and $\mathbf{a}(\theta_r, \phi_r) \in \mathbb{C}^{M_\mathtt{rad} \times 1}$ (radar) is defined in Appendix \ref{App1_Proof_BFGain_UpBound}, where $M_\mathtt{BS} = N^{(\mathtt{BS})}_\mathtt{az} \times N^{(\mathtt{BS})}_\mathtt{el}$ and $M_\mathtt{rad} = N^{(\mathtt{rad})}_\mathtt{az} \times N^{(\mathtt{rad})}_\mathtt{el}$.
\vspace{-10pt}
\subsection{Massive MIMO Downlink Beamforming Model}
Each massive MIMO cell has $K$ clusters/virtual sectors with mutually disjoint angular support. From each cluster, only one user is co-scheduled and served on the massive MIMO downlink using joint spatial division multiplexing (JSDM) \cite{JSDM_Adhikary_Caire_TIT_2013}. Hence at any given point of time, we assume that $K$ users are co-scheduled from $K$ different clusters. We consider a highly spatially correlated downlink channel, given by the one-ring model as $\mathbf{h_i} = \sqrt{\beta_i}\mathbf{U_i} \mathbf{\Lambda}^{1/2}_\mathbf{i} \mathbf{z_i} \in \mathbb{C}^{M_\mathtt{BS} \times 1}$ with channel covariance $\mathbf{R_i} = \mathbf{U_i} \mathbf{\Lambda}_\mathbf{i} \mathbf{U}^H_\mathbf{i}$ \cite{JSDM_Adhikary_Caire_TIT_2013}, where $\beta_i$ is the large-scale pathloss for the $i^{th}$ user, $\mathbf{U_i} \in \mathbb{C}^{M_\mathtt{BS} \times r_i}$ is the orthonormal matrix of eigenvectors, $\mathbf{\Lambda_i} \in \mathbb{R}^{r_i \times r_i}$ is the diagonal matrix of eigenvalues, and $\mathbf{z_i} \sim \mathcal{CN}(\mathbf{0, I_{r_i}}) \in \mathbb{C}^{r_i \times 1}$ is a complex Gaussian random vector, where $r_i \ll M_\mathtt{BS}$ is the rank of in the covariance matrix $\mathbf{R_i}$ in high spatially correlated downlink channel conditions \cite{JSDM_Adhikary_Caire_TIT_2013}. For simplicity, we consider that all users in the network have the same channel rank. The received signal $\mathbf{y} \in \mathbb{C}^{K \times 1}$ can be written as
\begin{align}
	\label{per_UE_rx_sig}
	\mathbf{y} & = \mathbf{H}^H \mathbf{W_{RF} W_{BB} d}  + \mathbf{n},
\end{align}
where $\mathbf{H} = [\mathbf{h_1}\ \mathbf{h_2}\ \cdots \mathbf{h_K}] \in \mathbb{C}^{M_\mathtt{BS} \times K}$ is the channel matrix, $\mathbf{W_{RF}} = [\mathbf{w_{RF,1}}\ \ \mathbf{w_{RF,2}}\cdots \mathbf{w_{RF,K}}] \in \mathbb{C}^{M_\mathtt{BS} \times K}$ is the RF beamformer that groups user clusters with disjoint angular support using nearly orthogonal beams, and $\mathbf{W_{BB}} = [\mathbf{w_{BB,1}}\ \cdots  \mathbf{w_{BB,K}}]$ $\in \mathbb{C}^{K \times K}$ is the baseband precoder \cite{JSDM_Adhikary_Caire_TIT_2013}. If the azimuth and elevation angular support of the $k^{th}$ user cluster is given by $\Theta_k = [\theta^{(\mathtt{min})}_k, \theta^{(\mathtt{max})}_k]$ and $\Phi_k = [\phi^{(\mathtt{min})}_k, \phi^{(\mathtt{max})}_k]$, then without loss of generality we consider that the RF beamformer is given by $\mathbf{w_{RF,k}} = \tfrac{1}{\sqrt{M_\mathtt{BS}}}\mathbf{a}(\theta_k, \phi_k)$, where $\theta_k = (\theta^{(\mathtt{min})}_k + \theta^{(\mathtt{max})}_k)/2$ and $\phi_k = (\phi^{(\mathtt{min})}_k + \phi^{(\mathtt{max})}_k)/2$. The data $\mathbf{d}= [d_1\ d_2\ \cdots\ d_K]^T \in \mathbb{C}^{K \times 1}$, such that $\mathbb{E}[\mathbf{d}]=\mathbf{0}$ and $\mathbb{E}[\mathbf{dd}^H] = \tfrac{P_{BS}}{K} \mathbf{I}$, where $d_k$ is the symbol intended for the $k^{th}$ UE and $P_{BS}$ is the total transmit power per BS. The noise $\mathbf{n} \in \mathbb{C}^{K \times 1}$ is spatially white with $\mathbf{n} \sim \mathcal{CN} (\mathbf{0}, \sigma^2_n \mathbf{I})$.
\begin{proposition}\label{Prop_WBB}
For the massive MIMO BS in the asymptotic regime, the baseband precoding matrix $\mathbf{W_{BB}} \approx \mathbf{I}$ for Zero-Forcing (ZF) and Maximum Ratio Transmission (MRT)\footnote{This result applies to other beamforming schemes such as Regularized Zero-Forcing (RZF) and Minimum Mean Square Error (MMSE). The proof involves additional steps but follows the same procedure presented below.}, when $K$ users from different clusters with mutually disjoint angular support are served. 
\end{proposition}
\begin{proof}
(Sketch) The MRT and ZF precoders are $\mathbf{W}^{(\mathtt{MRT})}_\mathbf{BB} = \mathbf{W}^H_\mathbf{RF} \mathbf{H}$ and
$\mathbf{W}^{(\mathtt{ZF})}_\mathbf{BB} = (\mathbf{H}^H \mathbf{W}_\mathbf{RF})^{-1}$ respectively. In the asymptotic regime $\mathbf{W}^H_\mathbf{RF} \mathbf{W_{RF}} \approx \mathbf{I}$ \cite{JSDM_Adhikary_Caire_TIT_2013}. For users in clusters with mutually disjoint angular support, $\mathbf{U}^H_\mathbf{i} \mathbf{w_{RB,j}} \approx \mathbf{0}, i \neq j$ \cite{JSDM_Adhikary_Caire_TIT_2013}. Therefore, $\mathbf{H}^H \mathbf{W_{RF}} \approx \mathbf{\Upsilon} = \text{diag}[\upsilon_1\ \upsilon_2 \cdots \upsilon_K]$. Since $\mathbb{E}[\mathbf{dd}^H] = \tfrac{P_{BS}}{K} \mathbf{I}$, when the per-user power constraint $\mathbb{E}[\| \mathbf{W_{RF} w_{BB,i} } d_i \|^2_2] = \tfrac{P_{BS}}{K}$ is imposed, we obtain the desired result.
\end{proof}
\begin{remark}
The above is true when $N^{(\mathtt{BS})}_\mathtt{el}, N^{(\mathtt{BS})}_\mathtt{az} \rightarrow \infty$. In the case of finite number of antenna elements, we consider a scheduler where the BS co-schedules $K$ users from clusters such that the above approximation is accurate. 
\end{remark}
\vspace{-10pt}
\subsection{Interference at the Radar due to a Single BS}
The radar is assumed to be searching/tracking a target above the horizon ($\phi < 0$) using a receive beamformer $\mathbf{w_{rad}} \in \mathbb{C}^{M_\mathtt{rad} \times 1}$. The interference signal prior to beamforming is $\mathbf{y_{rad}} = \mathbf{H}^H_{\mathbf{R}} \mathbf{W_{RF}} \mathbf{W_{BB}} \mathbf{d}$, where $\mathbf{H_R}$ is the high-${K_R}$ Rician channel between the BS and the radar from (\ref{Doub_dir_chan_model_f}). Upon receive beamforming, the interference signal is given by $i_\mathtt{rad} = \mathbf{w}^H_{\mathbf{rad}} \mathbf{H}^H_{\mathbf{R}} \mathbf{W_{RF} \mathbf{W_{BB}} d}$. Using equation (\ref{Doub_dir_chan_model_f}) and simplifying, we get
\begin{align}
i_\mathtt{rad} = &  \sqrt{\tfrac{\beta(d)}{K_R + 1}}\Big( \sqrt{K_R G_\mathtt{rad}(\theta_{r,L}, \phi_{r,L})} e^{-j \alpha_0} \mathbf{a}^H(\theta_{t,L}, \phi_{t,L}) +  \nonumber \\
& \sum\limits_{i=1}^{N_c} \sqrt{\tfrac{G_\mathtt{rad}(\theta_{r,i}, \phi_{r,i})}{ N_c}} \gamma'_i \mathbf{a}^H(\theta_{t,i}, \phi_{t,i}) \Big) \mathbf{W_{RF} W_{BB}d},
\end{align}
where $\gamma'_i = \gamma^*_i e^{-j \alpha_i}$, the radar beamforming gain $G_\mathtt{rad}(\theta_j, \phi_j) = |\mathbf{w}^H_\mathbf{rad} \mathbf{a}(\theta_j, \phi_j)|^2$, and $\alpha_0$ is the residual phase. The specular component can be ignored if $G_\mathtt{rad}(\theta_{r,L}, \phi_{r,L}) \gg G_\mathtt{rad}(\theta_{r,i}, \phi_{r,i})$. For a tractable worst-case analytical model, we make the following assumptions.
\begin{assumption}
(LoS beamforming gain dominance) The radar is scanning above the horizon with $\mathbf{w_{rad}} = \frac{\mathbf{a}(\theta_\mathtt{rad}, \phi_\mathtt{rad})}{\sqrt{M_\mathtt{rad}}}$ such that $G_\mathtt{rad}(\theta_{r,L}, \phi_{r,L}) \gg  G_\mathtt{rad}(\theta_{r,i}, \phi_{r,i})\ \forall\ 1\leq i \leq N_c$.
\end{assumption}
\begin{assumption}\label{BoresightAssumption}
(Boresight assumption) Boresight of the antenna array of each massive MIMO BS is aligned along the direction of radar ($\theta_{t,L}=0$)
as shown in Fig. \ref{Fig1_Radar_mMIMO_SpecShare_TopView}. 
\end{assumption}
\begin{assumption}
The cellular downlink is exactly co-channel with the radar system, and radar and cellular operating bandwidths are equal. Hence, the frequency-dependent rejection (FDR) factor of the radar is unity\footnote{The FDR is dependent on the radar architecture, interfering signal's spectrum, and is independent of other parameters. The interference power at the radar is inversely proportional to the FDR. Interested readers are referred to \cite{Hessar_Roy_Radar_WiFi_TAES_2016} for more details.}.
\end{assumption}
\begin{assumption}\label{Scheduler_support_assumption}
In each cell, the scheduler allocates resources to users in different clusters, where all but one cluster has disjoint angular support with the radar's azimuth w.r.t. the BS.
\end{assumption}
The boresight assumption is used for ease of exposition. As we will discuss in Section \ref{Sec_Int_at_Radar_mMIMO_Network} and Appendix \ref{App1_Proof_BFGain_UpBound}, the radar's azimuth w.r.t. the BS boresight does not impact the worst-case interference analysis. Based on the above assumptions, we have the following lemma.
\begin{lemma}\label{Lemma_DIUC}
The interference to the radar from each BS is only due to data transmissions towards a single cluster whose angular support overlaps with the boresight of the URA. 
\end{lemma}
\begin{proof}
Let the $K$ clusters have azimuth and elevation angles of support given by $\Theta_k$ and $\Phi_k$ respectively, for $1 \leq k \leq K$. In the asymptotic regime, if there is only one $k$ such that $\Theta_k \cap \{ 0^\circ \} \neq \emptyset$, then we get $\mathbf{a}^H(\theta_{t,L}, \phi_{t,L}) \mathbf{w_{RF,j}} \approx 0$ for $j \neq k$ and $\mathbf{a}^H(\theta_{t,L}, \phi_{t,L}) \mathbf{w_{RF,k}} \neq 0$ \cite{JSDM_Adhikary_Caire_TIT_2013}. The cluster that has its angular support overlapping with the BS boresight is termed as the ``Dominant Interfering User Cluster'' (DIUC).
\end{proof}

Based the above, we have the following key result. 
\begin{theorem}
The worst-case average interference power at the radar due to the DIUC is 
\begin{align}
\label{WorstCaseAvgIntPow_SingleBS}
\bar{I}_\mathtt{rad} < I^{(\mathtt{w})}_\mathtt{rad} = & \frac{\beta(d) G_\mathtt{rad}(\theta_\mathtt{rad}, \phi_\mathtt{rad},\theta_{r,L}, \phi_{r,L}) P_{BS}}{K} \cdot \nonumber \\
& \frac{|\mathbf{a}^H(0, \phi_{t,L}) \mathbf{a}(\theta_k, \phi_k)|^2 }{M_\mathtt{BS}}, 
\end{align}
where $G_\mathtt{rad}(\theta_\mathtt{rad}, \phi_\mathtt{rad},\theta_{r,L}, \phi_{r,L}) = \frac{|\mathbf{a}^H(\theta_\mathtt{rad}, \phi_\mathtt{rad})  \mathbf{a}(\theta_{r,L}, \phi_{r,L})|^2}{M_\mathtt{rad}}$.
\end{theorem}
\begin{proof}
Under the realistic assumption that each MPC is uncorrelated with the others, the average interference power $\bar{I}_\mathtt{rad} = \mathbb{E}[|i_\mathtt{rad}|^2]$ is given by
\begin{align}
\label{Irad_avgPow}
\bar{I}_\mathtt{rad} = & \tfrac{\beta(d)}{K_R + 1} \Big(K_R G_\mathtt{rad}(\theta_\mathtt{rad}, \phi_\mathtt{rad},\theta_{r,L}, \phi_{r,L}) \times \mathbb{E}[ \|\mathbf{a}^H(0, \phi_{t,L}) \cdot \nonumber \\
&  \mathbf{W_{RF} W_{BB} d}\|^2_2] +  \frac{1}{N_c} \sum_{i=1}^{N_c} G_\mathtt{rad}(\theta_\mathtt{rad}, \phi_\mathtt{rad},\theta_{r,i}, \phi_{r,i}) \cdot \nonumber \\
& \mathbb{E}[\gamma'^2_i \| \mathbf{a}^H(\theta_{t,i}, \phi_{t,i}) \mathbf{W_{RF} W_{BB} d}\|^2_2] \Big).
\end{align}
Using Assumption 1, we get 
\begin{align}
\bar{I}_\mathtt{rad} < & \beta(d) G_\mathtt{rad}(\theta_\mathtt{rad}, \phi_\mathtt{rad},\theta_{r,L}, \phi_{r,L}) \mathbb{E}[ \|\mathbf{a}^H(\theta_{t,L}, \phi_{t,L}) \cdot \nonumber \\
& \mathbf{W_{RF} W_{BB} d}\|^2_2], 
\end{align}
since $\mathbb{E}[|\gamma'_i|^2]=1$. In addition, by Proposition \ref{Prop_WBB}, Assumption 2 and Lemma \ref{Lemma_DIUC}, we get $\bar{I}_\mathtt{rad} < \mathbb{E}[|\mathbf{a}^H(0, \phi_{t,L}) \mathbf{w_{RF,k}} d_k|^2] \beta(d) G_\mathtt{rad}(\theta_\mathtt{rad}, \phi_\mathtt{rad},\theta_{r,L}, \phi_{r,L})$. Finally, using $\mathbb{E}[|d_k|^2]=P_{BS}/K$ and substituting the RF beamformer for the DIUC, we obtain the desired result.
\end{proof}
In summary, the worst-case average interference in high-$K_R$ Rician channels in the asymptotic regime resembles the Friis transmission equation, with the power scaled by the beamforming gains, and the power allocation factor to the DIUC. With this general result, we analyze the average interference due to the cellular network in the next section. 
\vspace{-1pt}
\section{Average Interference Power at the Radar due to the Cellular Downlink} \label{Sec_Int_at_Radar_mMIMO_Network}
Since the BS locations are modeled as a homogeneous PPP, the cells are polygons formed by the Poisson-Voronoi tessellation under maximum average power-based cell association \cite{andrews2016primer}. While the range of azimuth of a randomly selected point in the cell is independent of the cell size, the elevation angle depends on the cell size and hence, on $\lambda_\mathtt{BS}$. Compared to prior works \cite{Hessar_Roy_Radar_WiFi_TAES_2016}, \cite{Kim_WiFi_Radar_WCL_2017}, which focus on beamforming in the azimuth, mathematical modeling of elevation beamforming presents technical challenges due to (a) lack of radial symmetry in the PV cell, (b) possibility of arbitrarily large PV cells, and (c) correlation between the shapes and sizes of adjacent cells, which can affect the \textit{joint elevation distribution}. It is worthwhile to note that even though the presence of correlation hinders the analytical characterization of the \textit{worst-case interference distribution}, it does not impact the \textit{worst-case average interference}. However, the lack of radial symmetry and possibility of arbitrarily large cells need a more thoughtful treatment as far as average interference is concerned. To complicate matters further, the presence of sidelobes in the beamforming pattern makes it non-trivial to represent the worst-case beamforming gain as a function of the cell-size. Below, we develop the techniques to address these technical challenges, and present the worst-case and nominal average interference analysis. 
\begin{lemma}\label{Lemma_Monotonic_BF_Gain}
For a $N_\mathtt{az} \times N_\mathtt{el}$ URA with $\lambda/2$-spacing, if $\phi \in [-\pi/2, \pi/2), 0 \leq \phi_\mathtt{m} \leq \tfrac{\pi}{2}$, and $\theta \in [-\pi/2, \pi/2)$, then the upper bound of the beamforming gain is given by
\begin{align}
\label{BFGain_Tight_UpperBound}
G^{(\mathtt{max})}_\mathtt{BS}(\phi, \phi_\mathtt{m}) & = \underset{\substack{\phi_k\in [\phi_\mathtt{m},\pi/2) \\ \theta_k \in [-\pi/2, \pi/2)}}{\max} G_\mathtt{BS}(\theta, \phi, \theta_k, \phi_k)  \\
& = \begin{cases}
N_\mathtt{az}N_\mathtt{el}, \quad \quad \quad \quad \text{if } \phi_\mathtt{m} \leq \phi, \\
 G_\mathtt{BS} (0, \phi, 0, \phi_\mathtt{m}),\ \text{if } \sin \phi_\mathtt{m} \leq \tfrac{1 + N_\mathtt{el} \sin \phi}{N_\mathtt{el}} \\
 \tfrac{N_\mathtt{az}}{N_\mathtt{el}\sin^2 \big(\tfrac{\pi (\sin \phi_\mathtt{m} - \sin\phi )}{2}  \big)}, \quad \text{otherwise}
\end{cases} \nonumber
\end{align}
where $G_\mathtt{BS}(\theta, \phi, \theta_k, \phi_k) = \tfrac{1}{N_\mathtt{az} N_\mathtt{el}}|\mathbf{a}^H(\theta,\phi) \mathbf{a} (\theta_k, \phi_k)|^2$.
\end{lemma}
\begin{proof}
See Appendix \ref{App1_Proof_BFGain_UpBound}.
\end{proof}
We observe that the upper bound on $G_\mathtt{BS}$ is independent of the azimuth, since the maximum azimuth beamforming gain is universally upper bounded by $N_\mathtt{az}$. Hence, for ease of exposition, we use the boresight assumption as discussed in Assumption \ref{BoresightAssumption}.
\vspace{-10pt}
\subsection{Circumcircle-based Cell (CBC) Model}
To induce radial symmetry in the setup, the Voronoi cell needs to be modeled as a circle. When beamforming in the azimuthal direction of the radar, the worst-case interference occurs when the BS serving a user beamforms as close to horizon as possible, along which the radar is located. This corresponds to the scenario where the BS beamforms to the \textit{farthest point in the cell}, according to Lemma \ref{Lemma_Monotonic_BF_Gain}. Since the circumradius determines the distance to the farthest point in a cell, we propose a circumcircle-based construction as shown in Fig. \ref{Fig_Cell_Shape_Models}, with the following probability density function.

\begin{proposition}
The probability density function of the circumradius $r_c$ ($r_c > 0$) of a Poisson-Voronoi cell is 
\begin{align}
f_{R_C} (r_c) = & 8 \pi \lambda_\mathtt{BS} r_c e^{-4 \pi \lambda_\mathtt{BS} r^2_c} \Big[1 + \sum\nolimits_{k \geq 1} \Big\{ \tfrac{(-4 \pi \lambda_\mathtt{BS} r^2_c)^k}{k!} \cdot \nonumber \\ 
& \Big(\tfrac{\psi_k(r_c)}{8 \pi \lambda_\mathtt{BS} r_c} - \zeta_k (r_c) \Big) - \tfrac{(-4 \pi \lambda_\mathtt{BS} r^2_c)^{k-1} \zeta_k (r_c) }{(k-1)!}  \Big\} \Big],  \\
\zeta_k(r_c) = & \int\limits_{\|\mathbf{u} \|_1 = 1, u_i \in [0,1]} \Big[\prod_{i=1}^{k} F(u_i) \Big] e^{4 \pi \lambda_\mathtt{BS} r^2_c \sum\limits_{i=1}^k \int\limits_{0}^{u_i} F(t) {\rm d}t} {\rm d}\mathbf{u}, \nonumber \\
\psi_k (r) = & \tfrac{{\rm d} \zeta_k(r)}{{\rm d} r}, 
F(t) = \sin^2(\pi t) \mathbbm{1}(0\leq t \leq \tfrac{1}{2}) + \mathbbm{1}(t > \tfrac{1}{2}),\nonumber 
\end{align}
where $\mathbbm{1}(\cdot)$ denotes the indicator function.
\end{proposition}
\begin{proof}
The result is obtained by differentiating the CDF of the circumradius ($F_{R_C}(r_c)$) \cite{Parida_Dhillon_CBRS_Access_2017} w.r.t. $r_c$ using Leibniz's rule.
\end{proof}

Using $f_{R_C}(r_c)$ and Lemma \ref{Lemma_Monotonic_BF_Gain}, we obtain the upper bound on the average interference in the following key result.
\begin{theorem}\label{Circum_rad_Model}
The worst-case average interference at the radar is given by 
\begin{align}
\label{Boresight_Cell_Edge_BF_CECC}
\bar{I}_\mathtt{rad, c} = & \tfrac{\lambda_\mathtt{BS} P_{BS} PL(r_0)}{K} \int\limits_{-\tfrac{\pi}{2}}^{\tfrac{\pi}{2}} \int\limits_{r_{\mathtt{exc}}}^\infty \int\limits_{0}^{\infty} \tfrac{r G_\mathtt{rad}(\theta_\mathtt{rad}, \phi_\mathtt{rad},\theta_{r,L},-\phi_{t,L}(r) )}{(r^2 + (h_\mathtt{rad} - h_\mathtt{BS})^2 )^{\alpha/2}} \cdot \nonumber \\
& G^{(\mathtt{max})}_\mathtt{BS}(\phi_{t,L}(r), \phi_\mathtt{m}(r_c)) f_{R_C}(r_c) {\rm d} r_c {\rm d}r {\rm d}\theta_{r,L}, \nonumber \\
\phi_{t,L}(r) & = \tan^{-1} \big( \tfrac{h_\mathtt{BS} - h_\mathtt{rad}}{r} \big),
\phi_\mathtt{m}(r_c) = \tan^{-1} \big( \tfrac{h_\mathtt{BS}} {r_c}\big).
\end{align}
\end{theorem}
\begin{proof}
See Appendix \ref{App2_Proof_Worst_Case_Int}.
\end{proof}
\begin{corollary}\label{Corollary_Circum_Model}
The approximate worst-case average interference at the radar is given by
\begin{align}
\label{Approx_Circum_Model_avgInt}
\bar{I}^{(\mathtt{app})}_\mathtt{rad, c} = & \tfrac{\lambda_\mathtt{BS} P_{BS} PL(r_0) }{K(\alpha - 2) r^{\alpha - 2}_\mathtt{exc}} \Big[ \int_{-\tfrac{\pi}{2}}^{\tfrac{\pi}{2}}  G_\mathtt{rad}(\theta_\mathtt{rad}, \phi_\mathtt{rad},\theta_{r,L},0) {\rm d}\theta_{r,L} \Big] \cdot \nonumber \\
& \Big[\int_{0}^{\infty}  G^{(\mathtt{max})}_\mathtt{BS} (0, \phi_\mathtt{m}(r)) f_{R_C} (r) {\rm d}r \Big]. 
\end{align} \qedhere
\end{corollary}
\begin{proof}
Since $r \gg h_\mathtt{BS}$ and $r \gg h_\mathtt{rad}$, we have $\phi_{t,L}(r) = -\phi_{r,L}(r) \approx 0$, 
and $(r^2 + (h_\mathtt{BS} - h_\mathtt{rad})^2)^{\tfrac{\alpha}{2}} \approx r^{\alpha}$. Using these approximations in $\bar{I}_\mathtt{rad, c}$, grouping the integrands, and integrating over $r$ yields the desired result.
\end{proof}

\begin{figure}[t]
	\centering
	\includegraphics[width=3.4in]{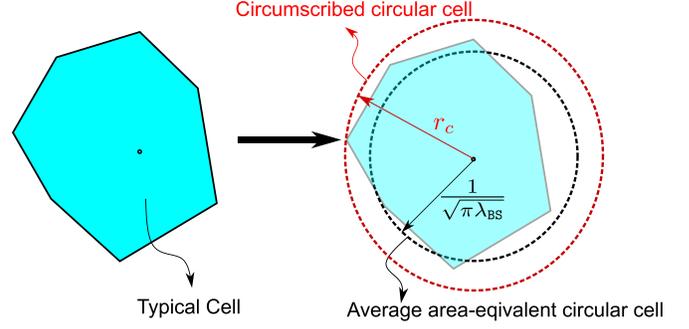}
	\caption{Radial symmetry can be induced by modeling the Voronoi cell as a (a) circumcircle, or (b) circle of area equal to that of the average typical cell.\\[-2ex]}
	\label{Fig_Cell_Shape_Models}
\end{figure}
\vspace{-15pt}
\subsection{Average Area-Equivalent Circular Cell (AAECC) Model}\label{AAEC_Model}
The circumcircle-based cell model results in a conservative value for average interference. A simpler, more optimistic model is to replace the Voronoi cell by a circle with an area equal to the average area of a typical cell given by $\tfrac{1}{\lambda_\mathtt{BS}}$. In this case, the cell radius $r_c = r_a = \tfrac{1}{\sqrt{\pi \lambda_\mathtt{BS}}}$, and the nominal average interference is given by the following theorem. 
\begin{theorem}\label{Theorem_AACC_Model}
The nominal mean and standard deviation of the interference power is
\begin{align}
\label{Boresight_Cell_Edge_BF_AACC_mean}
\bar{I}_\mathtt{rad, a}  = & \tfrac{\lambda_\mathtt{BS} P_{BS} PL(r_0)}{K} \int_{-\tfrac{\pi}{2}}^{\tfrac{\pi}{2}} \int_{r_{\mathtt{exc}}}^\infty \tfrac{r G_\mathtt{rad}(\theta_\mathtt{rad}, \phi_\mathtt{rad},\theta_{r,L},\phi_{r,L}(r) ) }{(r^2 + (h_\mathtt{rad} - h_\mathtt{BS})^2 )^{\alpha/2}} \cdot \nonumber \\
& G^{(\mathtt{max})}_\mathtt{BS}\big(\phi_{t,L}(r),\phi_\mathtt{m}(r_a) \big){\rm d}r {\rm d}\theta_{r,L}, \\
\label{Boresight_Cell_Edge_BF_AACC_var}
\sigma_\mathtt{rad, a} & = \tfrac{\sqrt{\lambda_\mathtt{BS}} P_{BS} PL(r_0)}{K} \Big[ \int_{-\tfrac{\pi}{2}}^{\tfrac{\pi}{2}} \int_{r_{\mathtt{exc}}}^\infty \tfrac{r G^2_\mathtt{rad}(\theta_\mathtt{rad}, \phi_\mathtt{rad},\theta_{r,L},\phi_{r,L}(r) )  }{(r^2 + (h_\mathtt{rad} - h_\mathtt{BS})^2 )^{\alpha}} \nonumber \\
& [G^{(\mathtt{max})}_\mathtt{BS}(\phi_{t,L}(r),\phi_\mathtt{m}(r_a) )]^2  {\rm d}r {\rm d}\theta_{r,L}\Big]^{\tfrac{1}{2}}.
\end{align}
\end{theorem}
\begin{proof}
This model is a special case of Theorem \ref{Circum_rad_Model}, where $f_{R_c} (r_c) = \delta \big( r_c - \tfrac{1}{\sqrt{\pi \lambda_\mathtt{BS}}} \big)$. Using the sifting property of the Dirac delta function $\delta(\cdot)$ in equation (\ref{Boresight_Cell_Edge_BF_CECC}), we obtain equation (\ref{Boresight_Cell_Edge_BF_AACC_mean}). The variance is obtained using Campbell's theorem, in a similar manner as Appendix \ref{App2_Proof_Worst_Case_Int}.
\end{proof}
\begin{corollary}\label{Corollary_AACC}
The approximate nominal average and variance of the interference power is
\begin{align}
\label{Approx_AvgInt_AACC}
\bar{I}^{(\mathtt{app})}_\mathtt{rad, a} = & \frac{\lambda_\mathtt{BS} P_{BS} PL(r_0) G^{(\mathtt{max})}_\mathtt{BS}\big(0,\phi_\mathtt{m}( r_a ) \big) }{K(\alpha - 2)r^{\alpha - 2}_\mathtt{exc}} \times \nonumber \\
& \int_{-\tfrac{\pi}{2}}^{\tfrac{\pi}{2}} G_\mathtt{rad}(\theta_\mathtt{rad}, \phi_\mathtt{rad},\theta, 0) {\rm d}\theta, \\
\label{Approx_std_dev_AACC}
\sigma^{(\mathtt{app})}_\mathtt{rad, a} = & \frac{\sqrt{\lambda_\mathtt{BS}} P_{BS} PL(r_0) G^{(\mathtt{max})}_\mathtt{BS}\big(0,\phi_\mathtt{m}( r_a ) \big) }
{\sqrt{(2 \alpha - 2)} K r^{\alpha - 1}_\mathtt{exc}}\times \nonumber \\  
& \sqrt{\int_{-\tfrac{\pi}{2}}^{\tfrac{\pi}{2}} G^2_\mathtt{rad}(\theta_\mathtt{rad}, \phi_\mathtt{rad},\theta, 0) {\rm d}\theta}.
\end{align}
\end{corollary}
\begin{proof}
The proof follows the same steps as Corollary \ref{Corollary_Circum_Model}.
\end{proof}
\vspace{-5pt}
\subsection{System Design Insights}
\subsubsection{Scaling of average interference power with BS density}
From (\ref{Boresight_Cell_Edge_BF_CECC}) and (\ref{Boresight_Cell_Edge_BF_AACC_mean}), we see that $\lambda_\mathtt{BS}$ impacts the average interference through the linear term and the BS beamforming gain ($G_\mathtt{BS}$) term. It is related to the cell size via the circumradius distribution and the average area of the typical cell, which impacts the \textit{minimum elevation angle} ($\phi_\mathtt{m}$). Note that this dependence is not observed in azimuth-only beamforming models. However, when $h_\mathtt{BS} \ll r_c$, the elevation angle $\phi_\mathtt{m}(r_c) \rightarrow 0$ and hence, $G_\mathtt{BS} \rightarrow M_\mathtt{BS}$. In this regime, the worst-case average interference power scales linearly with $\lambda_\mathtt{BS}$. 
\begin{comment} 
\subsubsection{Exclusion Zone Radius}
In practice, exclusion zones are defined based on the average aggregate interference power (for e.g. see \cite{FCC_3point5_GHz_Rules}). Using Corollaries \ref{Corollary_Circum_Model} and \ref{Corollary_AACC}, for an average interference threshold $\bar{I}_\mathtt{th}$ and $\alpha > 2$, the worst-case exclusion zone radius ($r^{(\mathtt{w})}_\mathtt{exc}$) can be obtained using
\begin{align*}
r^{(\mathtt{w})}_\mathtt{exc} & \approx \Big(\tfrac{\lambda_\mathtt{BS} P_{BS} PL(r_0) }{K(\alpha - 2) \bar{I}_\mathtt{th}} \Big[ \int_{-\tfrac{\pi}{2}}^{\tfrac{\pi}{2}}  G_\mathtt{rad}(\theta_\mathtt{rad}, \phi_\mathtt{rad},\theta,0) {\rm d}\theta \Big] \cdot \Big[\int_{0}^{\infty}  G^{(\mathtt{max})}_\mathtt{BS} (0,\phi_\mathtt{m}(r_c)) f_{R_C} (r) {\rm d}r \Big] \Big)^{\frac{1}{\alpha - 2}}. 
\end{align*}
\end{comment} 
\subsubsection{Constant Gap in Average Interference Predicted by CBC and AAECC Models}
By Corollaries (\ref{Corollary_Circum_Model}) and (\ref{Corollary_AACC}), we observe that the ratio of average interference powers is nearly independent of $r_\mathtt{exc}$, given by
\begin{align} 
\eta_\mathtt{ca} = \tfrac{\bar{I}^{(\mathtt{app})}_\mathtt{rad, c}}{\bar{I}^{(\mathtt{app})}_\mathtt{rad,a}}= \tfrac{\int\displaylimits_{0}^{\infty}  G^{(\mathtt{max})}_\mathtt{BS} (0, \phi_\mathtt{m}(r_c)) f_{R_C} (r_c) {\rm d}r_c}{G^{(\mathtt{max})}_\mathtt{BS}\big(0, \phi_\mathtt{m}\big( \tfrac{1}{\sqrt{\pi \lambda_\mathtt{BS}}} \big) \big)}. 
\end{align}
Note that $\eta_\mathtt{ca} \rightarrow 1$ when $h_\mathtt{BS} \sqrt{\pi\lambda_\mathtt{BS}} \rightarrow 0$, due to BS gain saturation.  

In the next section, we analyze the distribution of interference at the radar due to the massive MIMO cellular downlink. 
\vspace{-10pt}
\section{Distribution of Massive-MIMO Downlink Interference at the Radar}
To study the impact of large-scale network interference on aggregate radar performance metrics such as \textit{spatial probability of detection/false alarm} \cite{Chen_Dhillon_Liu_QoS_D2D_StochGeom_TCOM_2019}, deriving the distribution of interference due to spatial randomness in the BS locations is a key intermediate step. To accomplish this, a common approach in stochastic geometry literature is to characterize the Laplace transform of the interference distribution, which leverages the presence of an exponential term in Rayleigh fading channels \cite{andrews2016primer}. However in our case, the Laplace transform method is not applicable, since we ignore the small scale fading term in the high-$K_R$ Rician channel to model the worst-case interference scenario. To obtain useful results, we use the dominant interferer approximation  \cite{Parida_Dhillon_CBRS_Access_2017}, \cite{Vishnu_Dhillon_UAVs_TCOM_2017}, \cite{Heath_Kontouris_Bai_Dom_Int_TSP_2013}, \cite{Schloemann_Dhillon_Localze_TWC_2016} described below.

\begin{assumption}\label{Assump_dom_int_approx}
	In the cellular network, if the interference power due to the dominant interfering BS is $I_\mathtt{dom}$, and that due to the rest of the network is $I_\mathtt{rest}$, then the total interference power $(I_\mathtt{tot})$ is approximated by the sum of the dominant BS interference power and the average interference power due to the rest of the network, conditioned on the dominant interference power. Mathematically, it can be written as
	\begin{align}
	\label{equation:dom_int_method_defn}
	I_\mathtt{tot} \approx I_\mathtt{dom} + \mathbb{E}_{I_\mathtt{dom}} [I_\mathtt{rest}|I_\mathtt{dom}],
	\end{align} 
\end{assumption}
In the case of omnidirectional reception at the receiver, the distribution of $I_\mathtt{dom}$ is directly related to the distance distribution of the nearest transmitter in the point process \cite{Haenggi_dists_Rand_wir_ntwks_TIT_2005}, since the contour of equal interference power is a circle \cite{Heath_Kontouris_Bai_Dom_Int_TSP_2013}. However in our case, receive beamforming at the radar distorts radial symmetry, since received power depends on the azimuth and elevation angle, in addition to the distance from the interfering BS. Therefore, the first step is to characterize the contour curves of equal interference power, which is fundamental to calculating the void probability \cite{andrews2016primer} and hence, the distribution of $I_\mathtt{dom}$. In the rest of this paper, we assume cell-edge beamforming in the AAECC model to derive useful expressions for the interference distribution. In the following subsection, we characterize the equal interference contours in our radar-cellular coexistence scenario. 
\vspace{-10pt}
\subsection{Equal Interference Contours in Radar-Massive MIMO Spectrum Sharing}
The equal interference power contour $\mathcal{C}(I)$ contains points $(r,\theta)$ such that the received power due to a transmitter at location $(r,\theta) \in \mathcal{C}(I)$ is $I$. The following proposition denotes the contour lying outside the exclusion zone in the radar-cellular spectrum sharing scenario.
\begin{proposition}
	Under the AAECC model, the contour $\mathcal{C}(I_\mathtt{dom})$ is given by 
	\begin{align}
	\label{eq:contour_line_general_exp}
		\mathcal{C}(I_\mathtt{dom}) = & \Big \{ (r,\theta) \Big| r^{-\alpha}   G^{(\mathtt{max})}_\mathtt{BS} \big(-\phi(r), \phi_\mathtt{m}(1/\sqrt{\pi \lambda_\mathtt{BS}} \big) \big) \cdot \nonumber \\
		& G_\mathtt{rad} \big( \theta_\mathtt{rad}, \phi_\mathtt{rad}, \theta, \phi(r) \big) = \tfrac{K I_\mathtt{dom}}{PL(r_0) P_{BS}}, r \geq r_\mathtt{exc}, \nonumber \\
		& \theta \in \big[ -\tfrac{\pi}{2}, \tfrac{\pi}{2} \big] \Big \}, 
	\end{align}
	where  $\phi(r) = \tan^{-1} \Big( \frac{h_\mathtt{rad} - h_\mathtt{BS}}{r} \Big), \phi_\mathtt{m}(r')=\tan^{-1}( h_\mathtt{BS}/r')$.
\end{proposition}
\begin{proof}
	The worst-case interference power due to a massive MIMO BS at $(r,\theta)$ is given by (\ref{WorstCaseAvgIntPow_SingleBS}). Since the BSs inside the exclusion zone are inactive, the contour can be written as
	\begin{align}
	\label{eq:contour_line_general_exp_proof}
		\mathcal{C}(I_\mathtt{dom}) = & \Big \{ (r,\theta) \Big| \frac{PL(r_0) P_{BS} G_\mathtt{rad}(\theta_\mathtt{rad}, \phi_\mathtt{rad}, \theta, \phi(r))}{K r^{\alpha}} \cdot \nonumber \\
		& G^{(\mathtt{max})}_\mathtt{BS} (-\phi(r), \phi_\mathtt{m} (1/\sqrt{\pi \lambda_\mathtt{BS}})) = I_\mathtt{dom}, r \geq r_\mathtt{exc}, \nonumber \\
		& \theta \in \big[-\tfrac{\pi}{2}, \tfrac{\pi}{2} \big]  \Big \}. 
	\end{align}
Rearranging the terms independent of $(r,\theta)$, we obtain the desired result.
\end{proof}
In the case of large exclusion zone radii, we show in the following lemma that the equi-interference contour can be represented by the \textit{farthest distance between the contour and the radar}, when conditioned on the radar beamforming vector.
\begin{lemma}\label{lemma:Contour_large_exc_zone}
	Under the AAECC model, when $h_\mathtt{BS} \ll r_\mathtt{exc} \text{ and } h_\mathtt{rad} \ll r_\mathtt{exc}$, the equal interference contour is given by
	\begin{align}
		\label{eq:Contour_large_exc_zone_radius}
		\mathcal{C}(I_\mathtt{dom}) = & \Big \{ (r,\theta) \Big| r = \Big[ \tfrac{\sin \big(\tfrac{\pi}{2} N^{(\mathtt{rad})}_\mathtt{az} (\sin \theta_\mathtt{rad} \phi_\mathtt{rad} - \sin \theta) \big)}{N^{(\mathtt{rad})}_\mathtt{az} \sin \big( \tfrac{\pi}{2} (\sin \theta_\mathtt{rad} \phi_\mathtt{rad} - \sin \theta) \big)} \Big]^{2/\alpha} \times \nonumber \\
		& r_\mathtt{dom}, r_\mathtt{dom} \geq r_\mathtt{exc}, \theta \in \big[-\tfrac{\pi}{2}, \tfrac{\pi}{2} \big]  \Big \}, \nonumber \\
		\text{where } I_\mathtt{dom} \triangleq & \frac{P_{BS}PL(r_0)G^{(\mathtt{max})}_\mathtt{BS}(0, \phi_\mathtt{m}(1/\sqrt{\pi \lambda_\mathtt{BS}}))}{Kr^\alpha_\mathtt{dom}}\cdot \nonumber \\
		& \frac{N^{(\mathtt{rad})}_\mathtt{az} \sin^2 \big( \frac{\pi}{2} N^{(\mathtt{rad})}_\mathtt{el} \sin \phi_\mathtt{rad} \big)}{N^{(\mathtt{rad})}_\mathtt{el} \sin^2 \big( \tfrac{\pi}{2} \sin \phi_\mathtt{rad}  \big)}.
	\end{align}
\end{lemma}
\begin{proof}
	Please refer to Appendix \ref{appendix:proof_of_contour}.
\end{proof}
It is important to mention that the above formulation is rather unique to radar-cellular spectrum sharing scenarios, because of the large exclusion zones that are typically used. For other spectrum sharing scenarios in the absence of exclusion zones where coexisting users are located close to each other, this approximation is valid only when \textit{all the antenna heights are equal}. From equation (\ref{eq:Contour_large_exc_zone_radius}), we observe that there is a bijection between $r_\mathtt{dom}$, the farthest distance of the contour from the radar, and interference power $I_\mathtt{dom}$ under the AAECC model when $h_\mathtt{BS} \ll r_\mathtt{exc} \text{ and } h_\mathtt{rad} \ll r_\mathtt{exc}$, which are both reasonable assumptions in practice. Therefore, we can equivalently denote the \textit{equal interference contour} by $\mathcal{C}(r_\mathtt{dom})$, when conditioned on the radar beamforming vector. Fig. \ref{subfig:Illustration_of_contour_theta_minus_60} shows an example of the {equal interference contour}, which resembles a horizontal cross section of the radar's 3D beamforming pattern at elevation $\phi = 0^\circ$. In the following subsection, we derive the distribution of the dominant interference power $I_\mathtt{dom}$.

\begin{figure*}[t]
	\centering
	\begin{subfigure}[t]{0.48\textwidth}
		\raggedleft
		\includegraphics[width=3.2in]{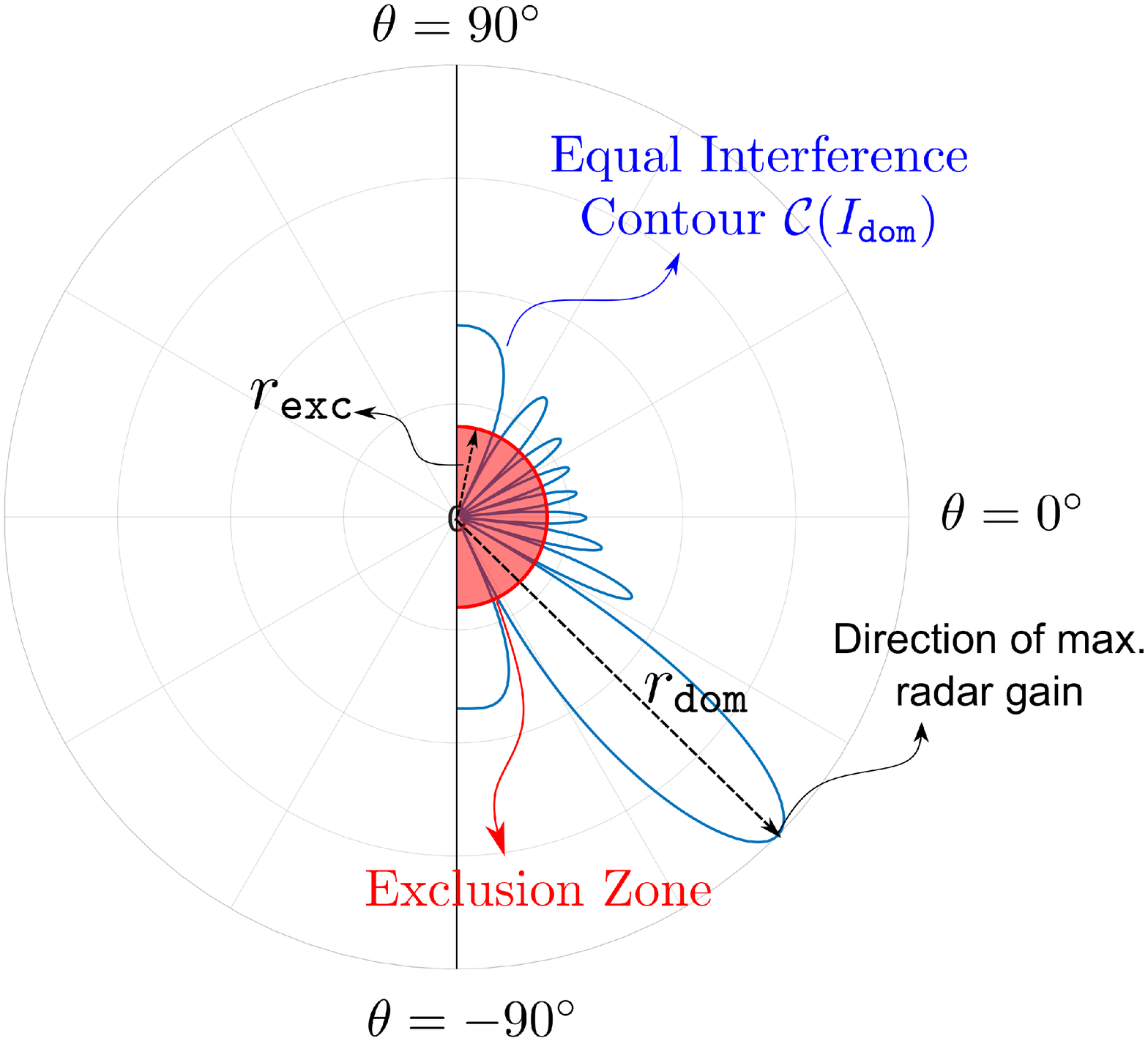}
		\caption{}
		\label{subfig:Illustration_of_contour_theta_minus_60}
	\end{subfigure}
	~
	\begin{subfigure}[t]{0.48\textwidth}
		\centering
		\includegraphics[width=3.2in]{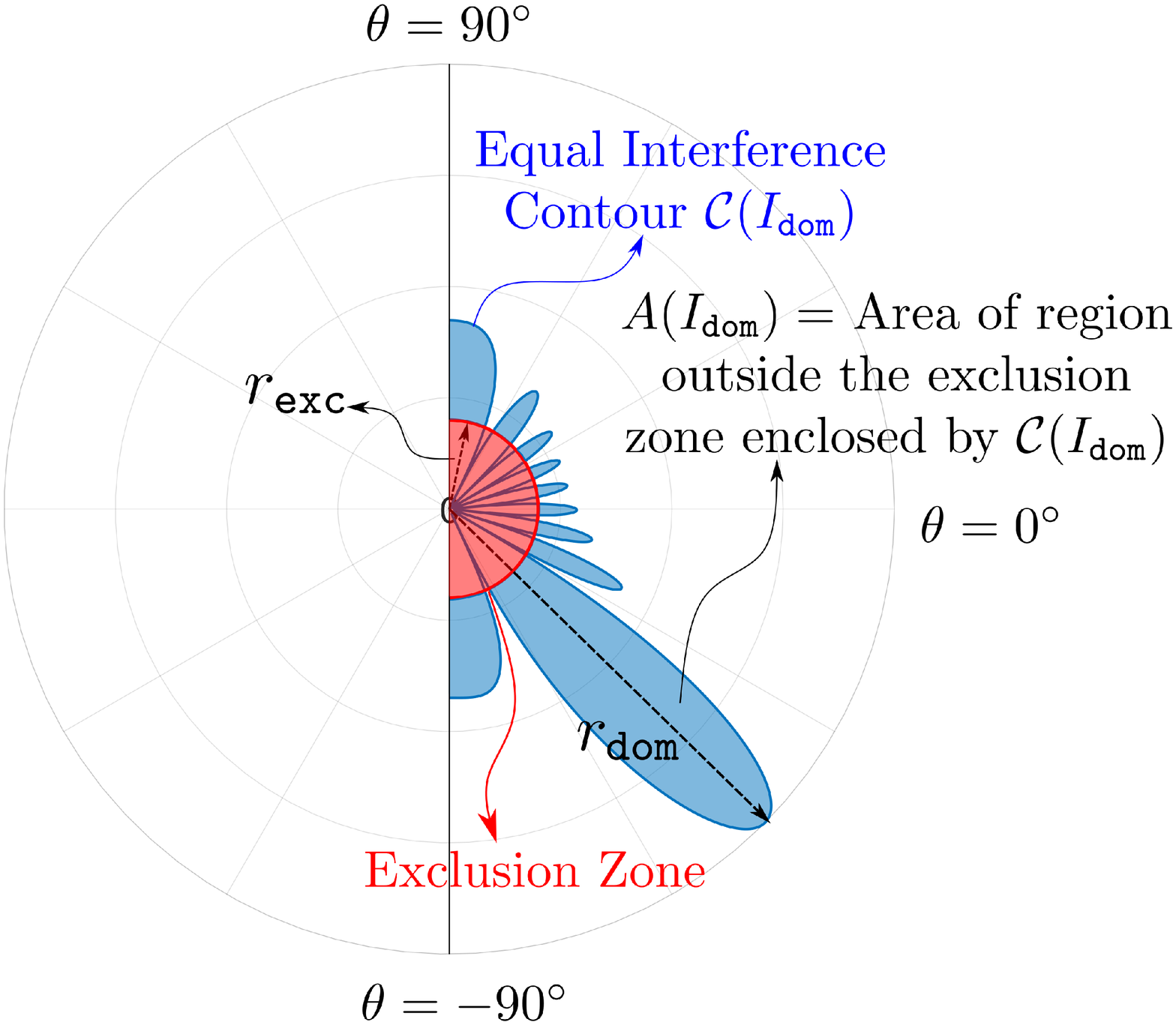}
		\caption{}
		\label{subfig:Illustration_of_contour_area_as_fn_of_r1}
	\end{subfigure}
	\caption{(a) Schematic of the equal interference power contour $\mathcal{C}(I_\mathtt{dom})$ in polar coordinates, for a radar with $N^{(\mathtt{rad})}_\mathtt{az} = N^{(\mathtt{rad})}_\mathtt{el} = 8$, scanning a target at $(\theta_\mathtt{rad}, \phi_\mathtt{rad}) = (-60^\circ, -5^\circ)$, with $\alpha = 3.5$, $r_\mathtt{exc} = 4$ km, and $r_1 = 20$ km. Distance of the farthest point on the contour is denoted by $r_\mathtt{dom}$. (b) Area of the region outside the exclusion zone but enclosed by $\mathcal{C}(I_\mathtt{dom})$ is denoted by $A(I_\mathtt{dom})$.}
	\label{fig:Contour_line_eq_interf_pow_example}
\end{figure*}
\vspace{-10pt}
\subsection{Distribution of $I_\mathtt{dom}$}
The distribution of $I_\mathtt{dom}$ is related to the void probability of a PPP in the region outside the exclusion zone enclosed by the equal interference contour \cite{andrews2016primer}, as shown in Fig. \ref{subfig:Illustration_of_contour_area_as_fn_of_r1}.  In the following key result, we derive an analytical expression for the area of this region $A(r_\mathtt{dom})$.
\begin{lemma}\label{lemma:area_B_r_closedform}
	Under the AAECC model, when $r_\mathtt{exc} \gg h_\mathtt{BS}$ and $r_\mathtt{exc} \gg h_\mathtt{rad}$, $A(r_\mathtt{dom})$ is given by
	\begin{align}
	\label{eq:B_r_dom_closedform}
	A(r_\mathtt{dom}) = & \frac{1}{2} \int_{-\frac{\pi}{2}}^{\frac{\pi}{2}} \max\Big( \Big[  \tfrac{\sin \big(\tfrac{\pi}{2} N^{(\mathtt{rad})}_\mathtt{az} (\sin \theta_\mathtt{rad} \cos \phi_\mathtt{rad} - \sin \theta) \big)}{N^{(\mathtt{rad})}_\mathtt{az} \sin \big( \tfrac{\pi}{2} (\sin \theta_\mathtt{rad} \cos \phi_\mathtt{rad} - \sin \theta) \big)} \Big]^\frac{4}{\alpha} \cdot  \nonumber \\
	& r^2_\mathtt{dom} , r^2_\mathtt{exc} \Big) {\rm d} \theta - \frac{\pi r^2_\mathtt{exc}}{2}, 
	\end{align} 
\end{lemma}
\begin{proof}
	Please refer Appendix \ref{appendix:proof_of_B_r_area}.
\end{proof}
Using the above result, the density of $r_\mathtt{dom}$ is characterized in the following lemma. 

\begin{lemma}\label{lemma:r_dom_density}
	The distribution and density function of $r_\mathtt{dom}$ are given by
	\begin{align}
	\label{eq:dist_of_r_dom}
	F_{R_\mathtt{dom}}(r_\mathtt{dom}) = & 1 - \exp \Big( - \frac{\lambda_\mathtt{BS}}{2} \int_{-\frac{\pi}{2}}^{\frac{\pi}{2}} \max\big(r^2_\mathtt{exc}, \tilde{r}^2_\mathtt{dom}(\theta) \big) {\rm d} \theta + \nonumber \\
	& \frac{\pi \lambda_\mathtt{BS} r^2_\mathtt{exc}}{2} \Big), \\
	\label{eq:density_of_r_dom_lemma}
	f_{R_\mathtt{dom}}(r_\mathtt{dom}) = & \lambda_\mathtt{BS} \Bigg[ \int_{-\tfrac{\pi}{2}}^{\tfrac{\pi}{2}} \Big[ \tfrac{\sin \big(\tfrac{\pi}{2} N^{(\mathtt{rad})}_\mathtt{az} (\sin \theta_\mathtt{rad} \cos \phi_\mathtt{rad} - \sin \theta) \big)}{N^{(\mathtt{rad})}_\mathtt{az} \sin \big( \tfrac{\pi}{2} (\sin \theta_\mathtt{rad} \cos \phi_\mathtt{rad} - \sin \theta) \big)} \Big]^{\frac{4}{\alpha}}  \cdot \nonumber \\
	& {r}_\mathtt{dom} \mathbbm{1} [\tilde{r}_\mathtt{dom}(\theta) \geq r_\mathtt{exc} ] {\rm d} \theta \Bigg] \cdot  \exp \Big(-\frac{\lambda_{BS}}{2} \times \nonumber \\
	& \int_{-\tfrac{\pi}{2}}^{\tfrac{\pi}{2}} \max\big(r^2_\mathtt{exc}, \tilde{r}^2_\mathtt{dom}(\theta) \big) {\rm d}\theta + \tfrac{\pi \lambda_{BS} r^2_\mathtt{exc}}{2}  \Big), 
	\end{align} 
	where $\tilde{r}_\mathtt{dom}(\theta) \triangleq r_\mathtt{dom} \Big[ \tfrac{\sin \big(\tfrac{\pi}{2} N^{(\mathtt{rad})}_\mathtt{az} (\sin \theta_\mathtt{rad} \cos \phi_\mathtt{rad} - \sin \theta) \big)}{N^{(\mathtt{rad})}_\mathtt{az} \sin \big( \tfrac{\pi}{2} (\sin \theta_\mathtt{rad} \cos \phi_\mathtt{rad} - \sin \theta) \big)} \Big]^{\frac{2}{\alpha}}$, and $\mathbbm{1}[\cdot]$ is the indicator function.
\end{lemma}
\begin{proof}
	Please refer Appendix \ref{appendix:lemma_r_dom_density_proof}.
\end{proof}
Since a bijection exists between $r_\mathtt{dom}$ and $I_\mathtt{dom}$, the density and distribution of $I_\mathtt{dom}$ can be derived similar to Lemma \ref{lemma:r_dom_density}, and is given in the following result.
\begin{lemma}\label{lemma:dist_and_density_of_I_dom}
	The distribution and density of $I_\mathtt{dom}$ under the AAECC model are given by
	\begin{align}
		\label{eq:I_dom_distrbn_final_expn}
		F_{I_\mathtt{dom}}(i_\mathtt{dom}) = & \exp\Big(-\tfrac{\lambda_\mathtt{BS} \kappa^{\frac{2}{\alpha}}}{2}\Big[ \int_{-\frac{\pi}{2}}^{\frac{\pi}{2}} \max \big( I^{-2/\alpha}_\mathtt{exc}, \tilde{i}^{-2/\alpha}_\mathtt{dom}(\theta)  \big) {\rm d} \theta \nonumber \\
		& - \pi I^{-2/\alpha}_\mathtt{exc} \Big] \Big), \\
		\label{eq:I_dom_density_final_expn}
		f_{I_\mathtt{dom}} (i_\mathtt{dom}) = & \tfrac{\lambda_\mathtt{BS} \kappa^{2/\alpha}}{\alpha} \Bigg[ \int_{-\frac{\pi}{2}}^{\frac{\pi}{2}} \Big[ \tfrac{\sin \big(\tfrac{\pi}{2} N^{(\mathtt{rad})}_\mathtt{az} (\sin \theta_\mathtt{rad} \cos \phi_\mathtt{rad} - \sin \theta) \big)}{N^{(\mathtt{rad})}_\mathtt{az} \sin \big( \tfrac{\pi}{2} (\sin \theta_\mathtt{rad} \cos \phi_\mathtt{rad} - \sin \theta) \big)} \Big]^{\frac{2}{\alpha}} \cdot \nonumber \\
		& i^{-(\alpha + 2)/\alpha}_\mathtt{dom}  \mathbbm{1}[\ \tilde{i}_\mathtt{dom} (\theta) \leq I_\mathtt{exc}] {\rm d} \theta \Bigg] \cdot \exp\Big(-\tfrac{\lambda_\mathtt{BS} \kappa^{\frac{2}{\alpha}}}{2} \cdot \nonumber \\
		& \Big[ \int_{-\frac{\pi}{2}}^{\frac{\pi}{2}} \max \big( I^{-\frac{2}{\alpha}}_\mathtt{exc}, \tilde{i}^{-\tfrac{2}{\alpha}}_\mathtt{dom}(\theta)  \big) {\rm d} \theta - \pi I^{-\frac{2}{\alpha}}_\mathtt{exc} \Big] \Big),
	\end{align}
	where $\tilde{i}_\mathtt{dom}(\theta) = i_\mathtt{dom} \tfrac{ N^{(\mathtt{rad})}_\mathtt{az} \sin \big( \tfrac{\pi}{2} (\sin \theta_\mathtt{rad} \cos \phi_\mathtt{rad} - \sin \theta) \big)}{\sin \big(\tfrac{\pi}{2} N^{(\mathtt{rad})}_\mathtt{az} (\sin \theta_\mathtt{rad} \cos \phi_\mathtt{rad} - \sin \theta) \big)}$, $\kappa = \frac{P_{BS}PL(r_0)G^{(\mathtt{max})}_\mathtt{BS}(0, \phi_\mathtt{m}(1/\sqrt{\pi \lambda_\mathtt{BS}}))}{K}\cdot \tfrac{N^{(\mathtt{rad})}_\mathtt{az} \sin^2 \big( \frac{\pi}{2} N^{(\mathtt{rad})}_\mathtt{el} \sin \phi_\mathtt{rad} \big)}{N^{(\mathtt{rad})}_\mathtt{el} \sin^2 \big( \tfrac{\pi}{2} \sin \phi_\mathtt{rad}  \big)}$, and $I_\mathtt{exc} = \kappa r^{-\alpha}_\mathtt{exc}$.
\end{lemma}
\begin{proof}
From equation (\ref{eq:Contour_large_exc_zone_radius}), we observe that the bijection between the dominant interference power $I_\mathtt{dom}$ and the corresponding farthest contour distance $r_\mathtt{dom}$ can be represented by $I_\mathtt{dom} = \kappa r^{-\alpha}_\mathtt{dom}$. Since $I$ monotonically decreases with increasing $r$, the CDF of $I_\mathtt{dom}$ is given by $\mathbb{P}[I_\mathtt{dom} \leq i_\mathtt{dom}] = \mathbb{P}[R_\mathtt{dom} \geq r_\mathtt{dom}]$. Using equation (\ref{eq:dist_of_r_dom}), we get $F_{I_\mathtt{dom}}(i_\mathtt{dom}) = \exp\big( - \lambda_\mathtt{BS} A(r_\mathtt{dom}) \big)$ for $r_\mathtt{dom} \geq r_\mathtt{exc}$. Using the bijection and simplifying, we get the desired CDF. The density is obtained in a similar manner as Lemma \ref{lemma:r_dom_density}, by differentiating equation (\ref{eq:I_dom_distrbn_final_expn}) w.r.t. $i_\mathtt{dom}$. 
\end{proof}
\vspace{-15pt}
\subsection{Total Interference Power at the Radar}
Since $r_\mathtt{dom}$ can equivalently represent the equal interference contour $\mathcal{C}(I_\mathtt{dom})$, we use Lemma \ref{lemma:r_dom_density} in the following result to approximately characterize the total interference power at the radar, using the dominant interferer method.
\begin{theorem}\label{theorem:I_tot_dom_int_method}
	The total interference power at the radar under the AAECC model and Assumption \ref{Assump_dom_int_approx} is given by
	\begin{align}
		\label{eq:domin_interf_approx_closed_form_rad_mMIMO}
		I_\mathtt{tot}(r_\mathtt{dom}) = & \kappa \Big[r^{-\alpha}_\mathtt{dom} + \tfrac{\lambda_\mathtt{BS}}{\alpha - 2} \int_{-\tfrac{\pi}{2}}^{\tfrac{\pi}{2}}  \big(\max\big(r_\mathtt{exc}, \tilde{r}_\mathtt{dom} (\theta) \big) \big)^{-\alpha + 2} \cdot \nonumber \\
		& \Big[ \tfrac{\sin \big(\tfrac{\pi}{2} N^{(\mathtt{rad})}_\mathtt{az} (\sin \theta_\mathtt{rad} \cos \phi_\mathtt{rad} - \sin \theta) \big)}{N^{(\mathtt{rad})}_\mathtt{az} \sin \big( \tfrac{\pi}{2} (\sin \theta_\mathtt{rad} \cos \phi_\mathtt{rad} - \sin \theta) \big)} \Big]^2 {\rm d}\theta \Big].
	\end{align}
\end{theorem}
\begin{proof}
	Please refer Appendix \ref{appendix:proof_of_I_tot_dom_int_apprx}.
\end{proof}
\begin{remark}\label{remark:finite_supp_I_tot_DI}
	It is worthwhile to note that $I_\mathtt{tot}(r_\mathtt{dom})$ has finite support, i.e. $I_\mathtt{tot} \in (0, I_\mathtt{exc} + \bar{I}_\mathtt{rad, a})$. This is because the maximum dominant interference power is upper bounded by $I_\mathtt{exc}$, and the corresponding conditional average interference power is $\bar{I}_\mathtt{rad,a}$ (equation \ref{Boresight_Cell_Edge_BF_AACC_mean}).
\end{remark}
In the following corollary, we prove that a bijection exists between $I_\mathtt{tot,DI}$ and $r_\mathtt{dom}$. 
\begin{corollary}\label{corollary:monoton_DI_I_tot}
	Under Theorem \ref{theorem:I_tot_dom_int_method}, $I_\mathtt{tot}$ monotonically decreases with $r_\mathtt{dom}$. 
\end{corollary}
\begin{proof}
	The proof follows by showing that both terms in equation (\ref{eq:domin_interf_approx_closed_form_rad_mMIMO}) monotonically decrease with $r_\mathtt{dom}$. It is clear that $I_\mathtt{dom}$ monotonically decreases with increasing $r_\mathtt{dom}$. In addition, we note that $\mathcal{A}(r_\mathtt{dom}) \subset \mathcal{A}(kr_\mathtt{dom})\ \forall\ k \in \mathbb{R}, k > 1$. As a result, the integration region and hence, the average interference power in equation (\ref{equation:Avg_est_condn_on_I_dom}) shrinks as $r_\mathtt{dom}$ increases. Therefore, the sum of these terms decreases monotonically with $r_\mathtt{dom}$.
\end{proof}
Hence, a bijection exists between $r_\mathtt{dom}$ and $I_\mathtt{tot}$ under the dominant interferer approximation. Unfortunately, the mapping from $I_\mathtt{tot}$ to $r_\mathtt{dom}$ cannot be expressed in closed-form. Hence for tractability, we use the distribution of $r_\mathtt{dom}$ in place of $I_\mathtt{tot}$ to characterize the radar performance metrics in the following section. 
\vspace{-15pt}
\section{Characterization of Spatial Probability of Detection and False Alarm}
In this section, we use the distribution of $r_\mathtt{dom}$ to characterize the impact of cellular interference on the radar's detection and false alarm performance in a target tracking scenario. 
\vspace{-15pt}
\subsection{Radar Received Signal Model}
In the presence of cellular interference and noise, the aggregate received signal depends on the presence or absence of a target at $(\theta_\mathtt{rad}, \phi_\mathtt{rad})$, when the radar performs receive beamforming using the weights $\mathbf{w_{rad}} = \tfrac{1}{\sqrt{M_\mathtt{rad}}}\mathbf{a}(\theta_\mathtt{rad}, \phi_\mathtt{rad})$. Denoting the received signal post-beamforming at time index $n$ is $y_\mathtt{rad}[n]$, we assume that the radar calculates the test statistic $P_\mathtt{rad} = \frac{1}{N}\sum_{n=1}^N |y_\mathtt{rad}[n]|^2$ in an estimation window of $N$ samples. Let $\mathcal{H}_0$ denote the hypothesis that there is no target, and $\mathcal{H}_1$ denote the hypothesis that there is a target. We assume that each BS transmits i.i.d. complex Gaussian signals, and noise is i.i.d. circularly symmetric complex Gaussian. In near-LoS channel conditions, when BSs transmit i.i.d. Gaussian signals, the aggregate interference signal is Gaussian distributed \textit{when conditioned on the BS locations} $\mathbf{\Phi_{int}}$. Thus, the received signal under each hypothesis can be written as

\begin{align}
	\label{equation:hypoth_test_radar1}
	\mathcal{H}_0 & : y_{\mathtt{rad},0}[n] = \sqrt{(I_\mathtt{tot}(\mathbf{\Phi_{int}}) + \sigma^2_w)} w[n], \\
	\mathcal{H}_1 & : y_{\mathtt{rad},1}[n] = \sqrt{(I_\mathtt{tot}(\mathbf{\Phi_{int}}) + \sigma^2_w)} w[n] + \sqrt{P_\mathtt{tar}} e^{j \alpha [n]},
	\label{equation:hypoth_test_radar2}
\end{align} 
where $I_\mathtt{tot}(\mathbf{\Phi_{int}})$ is the aggregate interference power, $\sigma^2_n$ denotes the noise variance, $w[n] \sim \mathcal{CN}(0,1)$, $P_\mathtt{tar}$ is the received power due to target scatter, and $\alpha[n]$ is the phase of the target return at time $n$. Using this system model, we have the following lemma. 

\begin{lemma}\label{lemma:central_non_central_chi_sq_dist}
	The conditional distributions of $P_\mathtt{rad}$ under $\mathcal{H}_0$ and $\mathcal{H}_1$ are given by
	\begin{align}
		\label{equation:dist_test_stat_und_2_hypoth}
		\mathcal{H}_0 : & F_{P_{\mathtt{rad},0}}(p|I_\mathtt{tot} (\mathbf{\Phi_{int}})) = \tfrac{1}{ (N-1)! } \gamma_l \Big( N, \tfrac{Np}{I_\mathtt{tot}(\mathbf{\Phi_{int}}) + \sigma^2_w} \Big), \nonumber \\
		\mathcal{H}_1 : & F_{P_\mathtt{rad},1}(p|I_\mathtt{tot} (\mathbf{\Phi_{int}})) = 1 - Q_N \Big( \sqrt{\tfrac{2 N P_\mathtt{tar}}{I_\mathtt{tot}(\mathbf{\Phi_{int}}) + \sigma^2_w}}, \nonumber \\
		& \sqrt{\tfrac{2 N p}{I_\mathtt{tot}(\mathbf{\Phi_{int}}) + \sigma^2_w}} \Big),
	\end{align} 
	where $\gamma_l (a,x) = \int_0^x z^{a-1} e^{-z} {\rm d}z$ is the lower incomplete gamma function, $Q_N (a,b) = \int_b^\infty z^N/a^{N-1} \cdot \exp(-(z^2 + a^2)/2) I_{N-1} (az) {\rm d}z$ is the Marcum Q-function, and $I_{N-1}(z)$ is the modified Bessel function of order $(N-1)$.
\end{lemma}
\begin{proof}
	We observe from equation (\ref{equation:hypoth_test_radar1}) that under hypothesis $\mathcal{H}_0$, each sample in the estimation window is i.i.d. complex Gaussian distributed such that
	\begin{align*}  \Re\Big(\frac{y_{\mathtt{rad},0}[n]\sqrt{2N}}{\sqrt{I_\mathtt{tot}(\mathbf{\Phi_{int}}) + \sigma^2_w}} \Big) & \sim \mathcal{N}(0, 1),  \text{ and } \\ \Im\Big(\frac{y_{\mathtt{rad},0}[n]\sqrt{2N}}{\sqrt{I_\mathtt{tot}(\mathbf{\Phi_{int}}) + \sigma^2_w}} \Big) & \sim \mathcal{N}(0, 1),
	\end{align*} 
	for $n=1,2,\cdots,N$, where $\Re(\cdot)$ and $\Im(\cdot)$ denote the real and imaginary parts. Taking the squared sum of these terms, we observe that $\frac{2N P_{\mathtt{rad},0}}{I_\mathtt{tot}(\mathbf{\Phi_{int}}) + \sigma^2_w}$ is chi-squared distributed with $2N$ degrees of freedom, and the CDF follows accordingly.
	
	Similarly, the received signal samples under $\mathcal{H}_1$ are independent such that
	\begin{align*} 
	\Re\Big(\frac{y_{\mathtt{rad},1}[n]\sqrt{2N}}{\sqrt{I_\mathtt{tot}(\mathbf{\Phi_{int}}) + \sigma^2_w}} \Big) & \sim \mathcal{N}\Big(\frac{\sqrt{2NP_\mathtt{tar}} \cos( \alpha [n]) }{\sqrt{I_\mathtt{tot}(\mathbf{\Phi_{int}}) + \sigma^2_w}}, 1 \Big), \nonumber \text{ and } \\
	\Im\Big(\frac{y_{\mathtt{rad},1}[n]\sqrt{2N}}{\sqrt{I_\mathtt{tot}(\mathbf{\Phi_{int}}) + \sigma^2_w}} \Big) & \sim \mathcal{N}\Big(\frac{\sqrt{2NP_\mathtt{tar}} \sin( \alpha [n]) }{\sqrt{I_\mathtt{tot}(\mathbf{\Phi_{int}}) + \sigma^2_w}}, 1 \Big), 
	\end{align*}
	for $n=1,2,\cdots,N$. Taking the squared sum of these terms, we see that $\frac{2NP_{\mathtt{rad},1}}{I_\mathtt{tot}(\mathbf{\Phi_{int}}) + \sigma^2_w}$ has a non-central chi-squared distribution with $2N$ degrees of freedom and non-central parameter $\lambda = \frac{2NP_\mathtt{tar}}{I_\mathtt{tot}(\mathbf{\Phi_{int}}) + \sigma^2_w}$. The CDF follows accordingly. 
\end{proof}
\begin{corollary} \label{corollary:CLT_approx_N_infty}
	When $N \rightarrow \infty$, the conditional distributions of $P_\mathtt{rad}$ under $\mathcal{H}_0$ and $\mathcal{H}_1$ become
	\begin{align}
		\label{equation:dist_test_stat_CLT_approx}
		\mathcal{H}_0 : & F_{P_{\mathtt{rad},0}} (p|I_\mathtt{tot} (\mathbf{\Phi_{int}})) = 1 - Q \Big( \tfrac{\sqrt{N} (p - I_\mathtt{tot}(\mathbf{\Phi_{int}}) - \sigma^2_w)}{I_\mathtt{tot}(\mathbf{\Phi_{int}}) + \sigma^2_w}\Big), \nonumber \\
		\mathcal{H}_1 : & F_{P_{\mathtt{rad},1}} (p|I_\mathtt{tot} (\mathbf{\Phi_{int}})) = \nonumber \\
		&  \quad 1 - Q \Big( \tfrac{\sqrt{N} (p - P_\mathtt{tar} - I_\mathtt{tot}(\mathbf{\Phi_{int}}) - \sigma^2_w)}{\sqrt{(P_\mathtt{tar} + I_\mathtt{tot}(\mathbf{\Phi_{int}}) + \sigma^2_w)^2 - P^2_\mathtt{tar}}} \Big),
	\end{align}
	where $Q(x) = 1/\sqrt{2 \pi} \int_x^\infty \exp(-u^2/2) {\rm d}u$ is the Q-function.
\end{corollary}
\begin{proof}
Observe that when $y_i \stackrel{\text{i.i.d.}}{\sim} \mathcal{CN}(0, \sigma^2), i=1,2,\cdots,N$ and $N \rightarrow \infty$, we have $\frac{1}{N} \sum_{i=1}^N |y_i|^2 \sim \mathcal{N}(\sigma^2, N^{-1} \sigma^4)$ \cite{Chen_Dhillon_Liu_QoS_D2D_StochGeom_TCOM_2019}. Hence, the CDF of $P_{\mathtt{rad},0}$ follows by replacing $\sigma^2$ by $\text{Var}(y_{\mathtt{rad},0}[n]) = I_\mathtt{tot}(\mathbf{\Phi_{int}}) + \sigma^2_w$.

On the other hand, the mean and variance of $|y_{\mathtt{rad},1}[n]|^2$ is finite and is given by $\mathbb{E}[|y_{\mathtt{rad},1}[n]|^2] = P_\mathtt{tar} + I_\mathtt{tot} + \sigma^2_w$ and  $\text{Var}(|y_{\mathtt{rad},1}[n]|^2) = (I_\mathtt{tot} + \sigma^2_w)^2 + 2 P_\mathtt{tar}(I_\mathtt{tot} + \sigma^2_w)$ respectively, for $n=1,2,\cdots,N$. Using the central limit theorem, the distribution of $P_{\mathtt{rad},1}$ approaches a Gaussian distribution with mean $\mathbb{E}[P_{\mathtt{rad},1}] = P_\mathtt{tar} + I_\mathtt{tot} + \sigma^2_w$ and variance $\text{Var}(P_{\mathtt{rad},1}) = N^{-1} [(I_\mathtt{tot} + \sigma^2_w)^2 + 2 P_\mathtt{tar}(I_\mathtt{tot} + \sigma^2_w)]$, when $N \rightarrow \infty$. The CDF follows accordingly.
\end{proof}

\subsection{Radar Performance Metrics}
When conditioned on the interference $I_\mathtt{tot}(\mathbf{\Phi_{int}})$, noise power $\sigma^2_n$, and the detection threshold $P_\mathtt{th}$, the probability of detection ($P_\mathtt{d}$) and false alarm ($P_\mathtt{fa}$) are calculated using
\begin{align}
 	\label{equation:P_fa_det_conditioned}
 	P_\mathtt{d} = \mathbb{P}[P_\mathtt{rad} > P_\mathtt{th}|\mathcal{H}_1, I_\mathtt{tot}(\mathbf{\Phi_{int}}), \sigma^2_w], \\
 	P_\mathtt{fa} = \mathbb{P}[P_\mathtt{rad} > P_\mathtt{th}|\mathcal{H}_0, I_\mathtt{tot}(\mathbf{\Phi_{int}}), \sigma^2_w].
\end{align}
We assume that the noise variance is constant. However, since the cellular downlink network is a PPP, we are interested in a spatially averaged variant of these probabilities. These are termed as the spatial detection probability ($\bar{P}_\mathtt{d}$), and the probability of spatial false alarm ($\bar{P}_\mathtt{fa}$), which are defined as \cite{Chen_Dhillon_Liu_QoS_D2D_StochGeom_TCOM_2019} 
\begin{align}
	\label{equation:prob_spatial_det_fa_defn}
	\bar{P}_\mathtt{d} = \int_{0}^\infty \mathbb{P}[P_\mathtt{rad} > P_\mathtt{th}|\mathcal{H}_1, I_\mathtt{tot} ] f_{I_{\mathtt{tot}}}(x) {\rm d}x, \\ 
	\bar{P}_\mathtt{fa} = \int_{0}^\infty \mathbb{P}[P_\mathtt{rad} > P_\mathtt{th}|\mathcal{H}_0, I_\mathtt{tot}] f_{I_{\mathtt{tot}}}(x) {\rm d}x.	
\end{align}
where $P_\mathtt{rad}$ is the test statistic, and $f_{I_\mathtt{tot}}(\cdot)$ is the density functions of the cellular interference power. For notational simplicity, the dependence of $I_\mathtt{tot}$ on the random BS locations ($\mathbf{\Phi_{int}}$) is omitted. In the following key result, we provided a tractable approximation to the spatial detection and false alarm probabilities. 

\begin{theorem}\label{theorem:P_fa_P_d_first_princip_DIA}
	Under Assumption \ref{Assump_dom_int_approx}, $\bar{P}_\mathtt{fa}$ and $\bar{P}_\mathtt{d}$ are given by
	\begin{align}
		\label{equation:prob_spatial_det_fa_approx_first_princip}
		\bar{P}_{\mathtt{fa}, \chi^2} = & 1 - \tfrac{1}{(N-1)!}\int_{r_\mathtt{exc}}^\infty  \gamma_l \Big(N, \tfrac{N P_\mathtt{th}}{I_\mathtt{tot,DI}(r_\mathtt{dom}) + \sigma^2_w} \Big) \cdot \nonumber \\
		& f_{R_\mathtt{dom}}(r_\mathtt{dom}) {\rm d}r_\mathtt{dom}, \nonumber \\
		\bar{P}_{\mathtt{d}, \chi^2} = & \int_{r_\mathtt{exc}}^\infty  Q_N \Big( \sqrt{\tfrac{2NP_\mathtt{tar}}{I_\mathtt{tot,DI}(r_\mathtt{dom}) + \sigma^2_w}}, \sqrt{\tfrac{2NP_\mathtt{th}}{I_\mathtt{tot,DI}(r_\mathtt{dom}) + \sigma^2_w}} \Big)\cdot \nonumber \\
		& f_{R_\mathtt{dom}}(r_\mathtt{dom}) {\rm d}r_\mathtt{dom},
	\end{align}
	where $f_{R_\mathtt{dom}}(\cdot)$ is the PDF of $r_\mathtt{dom}$ (equation (\ref{eq:density_of_r_dom_lemma})), and $I_\mathtt{tot,DI}$ is the total interference power under the dominant interferer approximation (equation (\ref{eq:domin_interf_approx_closed_form_rad_mMIMO})).
\end{theorem}
\begin{proof}
	Please refer Appendix \ref{appendix:proof_of_spat_Pd_Pfa_spat_chisq}. 
\end{proof}

\begin{corollary}\label{corollary:P_fa_P_d_first_princip_DIA}
	When $N \rightarrow \infty$, the probability of spatial detection and spatial false alarm under Assumption \ref{Assump_dom_int_approx} can be written as
	\begin{align}
		\label{equation:prob_spatial_fa_approx_CLT}
		\bar{P}_\mathtt{fa, CLT} = & \int_{r_\mathtt{exc}}^\infty  Q \Big( \tfrac{\sqrt{N}(P_\mathtt{th} - I_\mathtt{tot,DI}(r_\mathtt{dom}) - \sigma^2_w)}{I_\mathtt{tot,DI}(r_\mathtt{dom}) + \sigma^2_w} \Big)f_{R_\mathtt{dom}}(r_\mathtt{dom}) {\rm d}r_\mathtt{dom}, \\
		\label{equation:prob_spatial_d_approx_CLT}
		\bar{P}_\mathtt{d,CLT} = & \int_{r_\mathtt{exc}}^\infty  Q \Big( \tfrac{\sqrt{N}(P_\mathtt{th} - P_\mathtt{tar} - I_\mathtt{tot,DI}(r_\mathtt{dom}) - \sigma^2_w)}{\sqrt{(P_\mathtt{tar} + I_\mathtt{tot,DI}(r_\mathtt{dom}) + \sigma^2_w)^2 - P^2_\mathtt{tar}}} \Big) \cdot \nonumber \\
		& f_{R_\mathtt{dom}}(r_\mathtt{dom}) {\rm d}r_\mathtt{dom}.
	\end{align}
\end{corollary}
\begin{proof}
	The proof is similar to Theorem \ref{theorem:P_fa_P_d_first_princip_DIA}, and follows from the complementary CDF of the Gaussian distribution in Corollary \ref{corollary:CLT_approx_N_infty}. 
\end{proof}

\begin{figure*}[t]
	\centering
	\begin{subfigure}[t]{0.48\textwidth}
		\centering
		\includegraphics[width=3.2in]{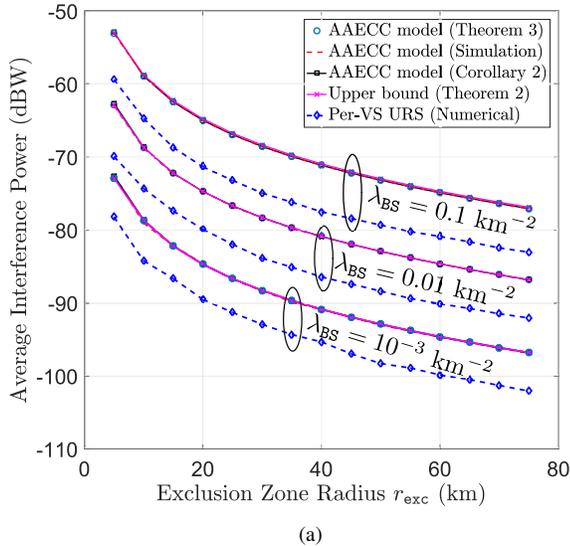}
		\caption{}
		\label{subfig:Worst_case_avg_all}
	\end{subfigure} 
	~
	\begin{subfigure}[t]{0.48\textwidth}
		\centering
		\includegraphics[width=3.2in]{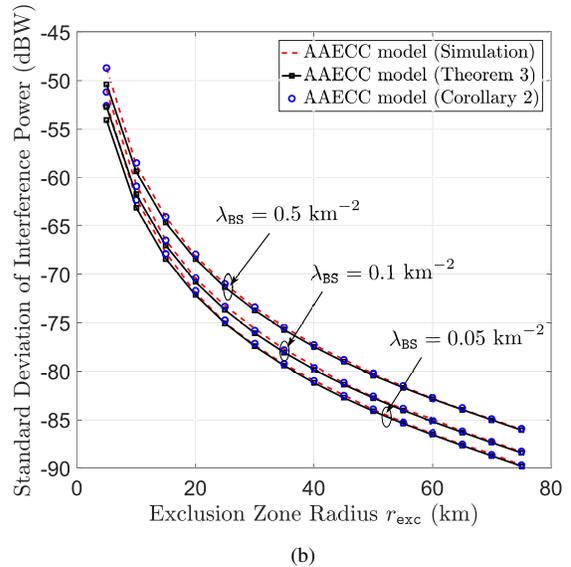}
		\caption{}
		\label{subfig:AAECC_std_dev}
	\end{subfigure} 
	\caption{(a) Worst-case average interference power at the radar under the AAECC and CBC models, as a function of $r_\mathtt{exc}$. (b) Standard deviation of worst-case interference power at the radar under the AAECC model. Base station densities $\lambda_\mathtt{BS}=0.05, 0.1,0.5$ (km$^{-2}$), $h_\mathtt{rad} = 20$ m, $h_\mathtt{BS} = 50$ m, $N^{(\mathtt{BS})}_\mathtt{az}= N^{(\mathtt{BS})}_\mathtt{el} = 10, N^{(\mathtt{rad})}_\mathtt{az} = N^{(\mathtt{rad})}_\mathtt{el} = 40, \theta_\mathtt{rad} = 60^\circ, \phi_\mathtt{rad} = -10^\circ.$} 
	\label{Fig_WorstCaseInt_all_Globecom}
\end{figure*}
\vspace{-8pt}
\section{Numerical Results and Discussion}\label{Sec_Numerical_Results}
In this section, we validate our theoretical results using Monte-Carlo simulations. We consider a typical radar operating at $f_c = 5$ GHz, located at the origin equipped with a $N^{(\mathtt{rad})}_\mathtt{az}\times N^{(\mathtt{rad})}_\mathtt{el}$ URA, mounted at a height of $h_\mathtt{rad} = 20$ m. The radar is scanning a region above the horizon at $(\theta_\mathtt{rad}, \phi_\mathtt{rad}) = (60^\circ, -10^\circ)$. 
The BSs are distributed as a PPP, with varying intensities. Each massive MIMO BS is co-channel with the radar, and is equipped with a $N^{(\mathtt{BS})}_\mathtt{az}\times N^{(\mathtt{BS})}_\mathtt{el}$ URA deployed at a height of $h_\mathtt{BS} = 50$ m. The circular exclusion zone around the radar has a \textit{minimum radius} of $r^{(\mathtt{min})}_\mathtt{exc} = 5$ km. The boresight of each massive MIMO BS URA is aligned along the direction of the radar ($\theta_k=0$ in the LCS). 
In each cell, the massive MIMO BS transmits a total power of $P_{BS} = 1$ W, equally allocated among co-scheduled UEs from $K=4$ clusters with mutually disjoint angular support. To model the pathloss in the downlink and the BS to radar channels, we assume the 3GPP 3D Urban Macro (3D UMa) LoS pathloss model \cite{3GPP5GNR_ChanModels},
\begin{align*}
PL(d) = P(h_\mathtt{BS}, h_\mathtt{rad}) + 20 \log_{10}(f_c) + 40 \log_{10}(d) \quad \text{(dB)}, \nonumber \\
P(h_\mathtt{BS}, h_\mathtt{rad}) = 28 - 9 \log_{10}((h_\mathtt{BS} - h_\mathtt{rad})^2) \quad \text{(dB)},
\end{align*}
where $f_c$ (GHz), and $d$ (m).

\begin{table}[t]
	\renewcommand{\arraystretch}{1.0}
	\caption{Approximate Values of $\eta_\mathtt{ca}$\\
	[-5ex]}
	\label{Tab_Ratio_Vals}
	\centering
	\footnotesize
	\begin{tabular}{|c|c|c|c|c|c|c|}
		\hline
		$h_\mathtt{BS}\sqrt{\pi \lambda_\mathtt{BS}}$ & 0.0089 & 0.0198 & 0.028 & 0.044 & 0.0886 & 0.1253\\
		\hline
		$\eta_\mathtt{ca}$ & 1.004 & 1.022 & 1.045 & 1.254 & 1.608 & 2.905 \\
		\hline
	\end{tabular} \\
[-2ex]
\end{table}
\subsection{Comparison of Worst-Case Interference under CBC and AAECC Models}
Fig. \ref{subfig:Worst_case_avg_all} shows the average interference power derived in Section \ref{Sec_Int_at_Radar_mMIMO_Network} under different cell models, as a function of exclusion zone radius for different BS intensities. For comparing the wost-case interference models with conventional cellular downlink schedulers, we also consider the \textit{Per-Virtual Sector Uniformly Random Scheduler (Per-VS URS)} scheme. We consider a PPP of user equipments (UEs) with intensity $\lambda_\mathtt{UE} = 100 \lambda_\mathtt{BS}$, served by the nearest massive MIMO BS. In the Per-VS URS scheme, the BS co-scheduled one UE uniformly at random from each of the $K=4$ user clusters. We use the one-ring model \cite{JSDM_Adhikary_Caire_TIT_2013} to generate the spatial downlink channel, with a scatterer ring radius of $30$ m around each UE, and Assumptions \ref{BoresightAssumption} and \ref{Scheduler_support_assumption} are relaxed. 

We observe that the average interference power under the Per-VS URS model differs from the upper bound by $\sim 6$ dB, and shows the same trends as the AAECC model and the upper bound, despite the differences in the scheduling scheme (beamforming towards cell-edge user beamforming vs. users in each VS chosen uniformly at random). We observe that the upper bound is remarkably tight w.r.t. the AAECC model, especially for $\lambda_\mathtt{BS} \leq 0.1$. For reference, we also plot the approximate average interference power from Corollary \ref{Corollary_AACC}. The approximately linear scaling of average interference power with $\lambda_\mathtt{BS}$ can also be observed, since the average interference power drops by $\approx 10$ dB when $\lambda_\mathtt{BS}$ is decreased by an order of magnitude.

We also observe that the ratio of average interference powers $\eta_\mathtt{ca}$ is approximately constant, and is tabulated for the \textit{elevation parameter} $h_\mathtt{BS} \sqrt{\pi \lambda_\mathtt{BS}}$ in Table \ref{Tab_Ratio_Vals}. For 3GPP UMa deployments with inter-site distance $r_\mathtt{ISD}$, the typical $h_\mathtt{BS}/r_\mathtt{ISD}=0.05$ \cite{3GPP5GNR_ChanModels}. The corresponding $h_\mathtt{BS}\sqrt{\pi \lambda_\mathtt{BS}}=0.095$, for which $2 \text{ dB} < \eta_\mathtt{ca} < 4.6 \text{ dB}$ as seen in Table \ref{Tab_Ratio_Vals}. Thus the bound is remarkably tight, which makes it useful for worst-case analysis of practical radar-5G NR spectrum sharing deployments. It is worthwhile to mention that in order to analyze the interference under the per-VS URS scheme, a key intermediate step is to derive the joint azimuth-elevation distribution of associated UEs in a cell. However due of the radial asymmetry of the PV cell, this problem is intractable (please refer \cite{mankar2020distance} for an example). This is the motivation for analyzing the worst-case scenario under a cell-edge beamforming model, which avoids the need for modeling the users' spatial distribution.

Fig. \ref{subfig:AAECC_std_dev} shows the standard deviation of the interference power under the AAECC model as a function of $r_\mathtt{exc}$ and $\lambda_\mathtt{BS}$. We observe that the accuracy of our theoretical expressions follow similar trends as Fig. \ref{subfig:Worst_case_avg_all}. However, we observe that the standard deviation decays faster than the average. In particular, comparing the slopes of the curves in Fig. \ref{subfig:Worst_case_avg_all} and \ref{subfig:AAECC_std_dev}, we observe that the standard deviation in Fig. \ref{subfig:AAECC_std_dev} decreases by $\approx 9$ dB per octave along $r_\mathtt{exc}$, in contrast to the average in Fig. \ref{subfig:Worst_case_avg_all} that decreases by $\approx 6$ dB per octave. This matches the scaling behavior predicted by Corollary \ref{Corollary_AACC}, where $\bar{I}_\mathtt{rad, a} \propto r^{-\alpha + 2}_\mathtt{exc}$ whereas $\sigma_\mathtt{rad, a} \propto r^{-\alpha + 1}_\mathtt{exc}$.

\begin{figure}[!t]
	\centering
	\begin{subfigure}[t]{0.48\textwidth}
		\centering
		\includegraphics[width=3.2in]{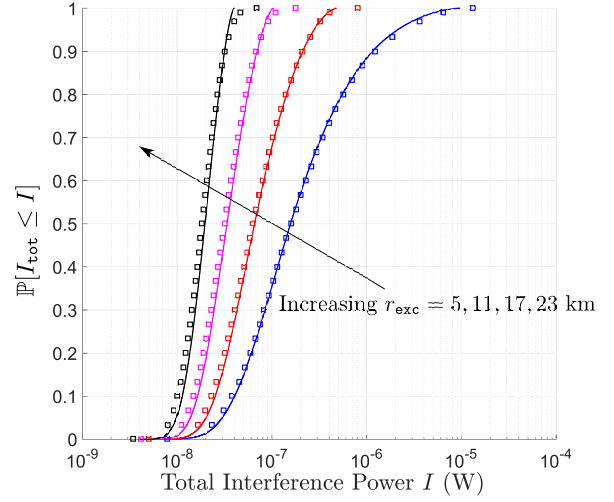}
		\caption{}
		\label{subfig:CDF_I_tot_dom_vs_num}
	\end{subfigure} 
	~
	\begin{subfigure}[t]{0.48\textwidth}
		\centering
		\includegraphics[width=3.2in]{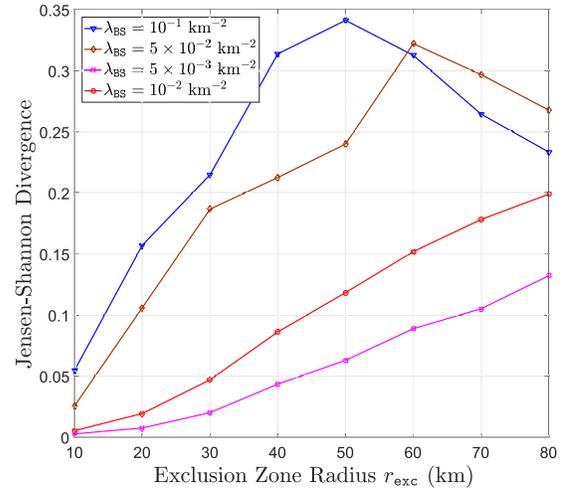}
		\caption{}
		\label{subfig:JS_div_I_tot_dom_vs_num}
	\end{subfigure} 
	\caption{(a) Distribution of total interference power ($I_\mathtt{tot}$) for different exclusion zone radii. Markers and solid lines represent the simulation and theoretical (Theorem \ref{theorem:I_tot_dom_int_method}) results respectively. (b) Jensen-Shannon divergence \cite{Lin_Div_measures_Shannon_TIT_1991} between the theoretical (Theorem \ref{theorem:I_tot_dom_int_method}) and numerical interference distributions, as a function of $\lambda_\mathtt{BS}$ and $r_\mathtt{exc}$. $\lambda_\mathtt{BS} = 0.01$ km$^{-2}$, $h_\mathtt{rad} = 20$ m, $h_\mathtt{BS} = 50$ m, $N^{(\mathtt{BS})}_\mathtt{az}= N^{(\mathtt{BS})}_\mathtt{el} = 10, N^{(\mathtt{rad})}_\mathtt{az} = N^{(\mathtt{rad})}_\mathtt{el} = 10, \theta_\mathtt{rad} = 60^\circ, \phi_\mathtt{rad} = -10^\circ.$} 
	\label{figure:dom_int_comparison}
\end{figure}
\subsection{Distribution of Total Interference Power}
Fig. \ref{subfig:CDF_I_tot_dom_vs_num} shows the distribution of total interference power for different $r_\mathtt{exc}$. Similar to the observations in Fig. \ref{subfig:AAECC_std_dev}, we observe that the distribution concentrates in narrower intervals around the average with increasing $r_\mathtt{exc}$. Overall, the analytical expression in Theorem \ref{theorem:I_tot_dom_int_method} obtained using the dominant interferer approximation matches well with the numerical results. However, we observe that as $r_\mathtt{exc}$ increases, there is a slight deviation in the CDF's upper tail. This is because the support of the actual interference is $[0, \infty)$, whereas that under the dominant interferer method is finite, as discussed in Remark \ref{remark:finite_supp_I_tot_DI}. 
Fig. \ref{subfig:JS_div_I_tot_dom_vs_num} shows the Jensen-Shannon divergence\footnote{Jensen-Shannon divergence (JSD) is a symmetrized version of the Kullback-Leibler divergence (KLD), and measures the similarity between two or more distributions. Unlike the KLD, the JSD guaranteed to be finite, and lies between 0 and 1 \cite{Lin_Div_measures_Shannon_TIT_1991}.} \cite{Lin_Div_measures_Shannon_TIT_1991}, which compares the similarity between the theoretical and simulation results in Fig. \ref{subfig:CDF_I_tot_dom_vs_num}. Lower JSD values imply that the distributions are similar. We observe that for a fixed $\lambda_\mathtt{BS}$, the JSD initially increases with $r_\mathtt{exc}$, and then decreases. This behavior can be explained as follows. 
\begin{enumerate} 
\item For very low $r_\mathtt{exc}$, Theorem \ref{theorem:I_tot_dom_int_method} is accurate since $I_\mathtt{dom} \gg \mathbb{E}[I_\mathtt{rest}|I_\mathtt{dom}]$, resulting in a close match and hence, a low JSD. 
\item For intermediate $r_\mathtt{exc}$, the JSD increases since $I_\mathtt{dom}$ and $\mathbb{E}[I_\mathtt{rest}|I_\mathtt{dom}]$ are comparable, thus degrading the accuracy of Theorem \ref{theorem:I_tot_dom_int_method}. 
\item For very high $r_\mathtt{exc}$, $I_\mathtt{dom} \ll \mathbb{E}[I_\mathtt{rest}|I_\mathtt{dom}]$ and $F_{R_\mathtt{dom}}(r) \rightarrow U(r - r_\mathtt{exc})$, as seen in (\ref{eq:dist_of_r_dom}). As a result, $I_\mathtt{dom} \rightarrow I_\mathtt{exc}$ for very large $r_\mathtt{exc}$, for which $I_\mathtt{tot} \rightarrow \mathbb{E}[I_\mathtt{rest}|I_\mathtt{exc}] = \bar{I}_\mathtt{rad,a}$ as discussed in Remark \ref{remark:finite_supp_I_tot_DI}. On the other hand, the true distribution converges to the average as well, since $\frac{\sigma^{(\mathtt{app})}_\mathtt{rad, a}}{\bar{I}^{(\mathtt{app})}_\mathtt{rad, a}} \propto \frac{1}{\sqrt{\lambda_\mathtt{BS} r^2_\mathtt{exc}}} \rightarrow 0$ (using (\ref{Approx_AvgInt_AACC})-(\ref{Approx_std_dev_AACC})) when $r_\mathtt{exc} \rightarrow \infty$. Since both distributions converge to the average interference power, the JSD tends to 0 when $r_\mathtt{exc}\rightarrow \infty$.  
\end{enumerate}
The same trends hold when $\lambda_\mathtt{BS}$ increases, which can be understood by similar arguments.

\begin{figure*}[t]
	\centering
	\begin{subfigure}[t]{0.48\textwidth}
		\centering
		\includegraphics[width=2.8in]{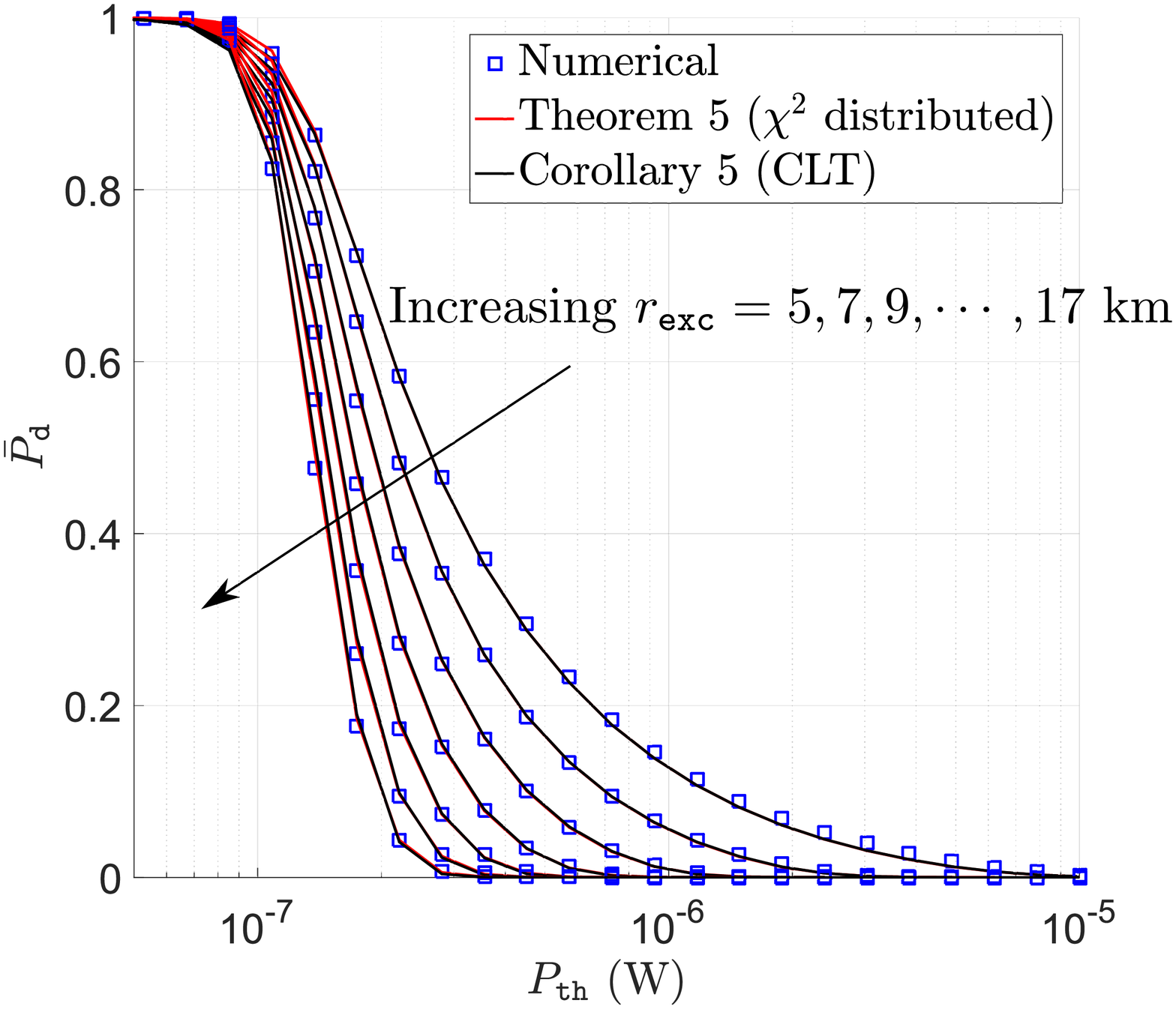}
		\caption{}
		\label{subfig:Spat_d_10milli}
	\end{subfigure}
	~
	\begin{subfigure}[t]{0.48\textwidth}
		\centering
		\includegraphics[width=2.8in]{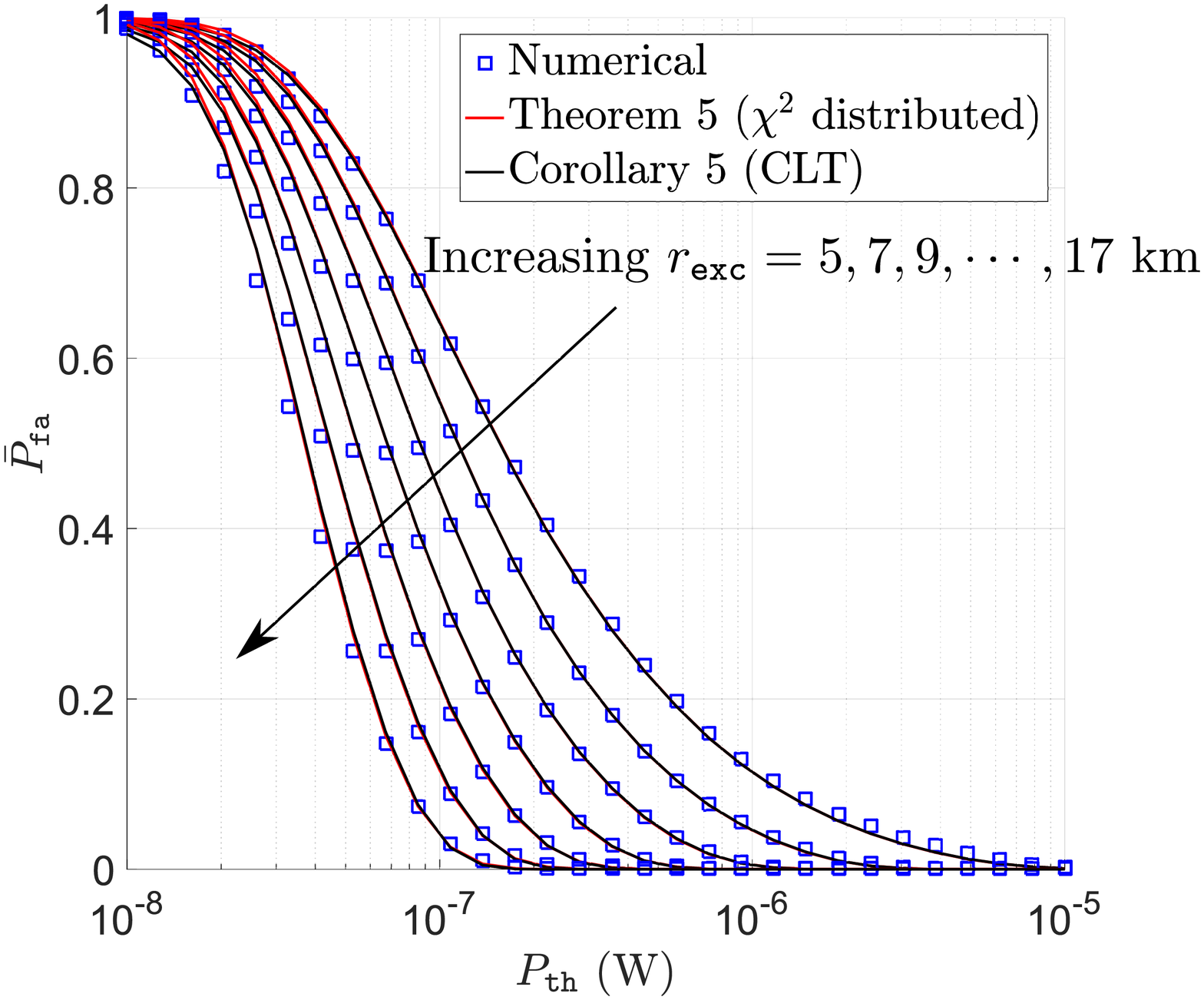}
		\caption{}
		\label{subfig:Spat_fa_10milli}		
	\end{subfigure}
	~
	\begin{subfigure}[t]{0.48\textwidth}
		\centering
		\includegraphics[width=2.8in]{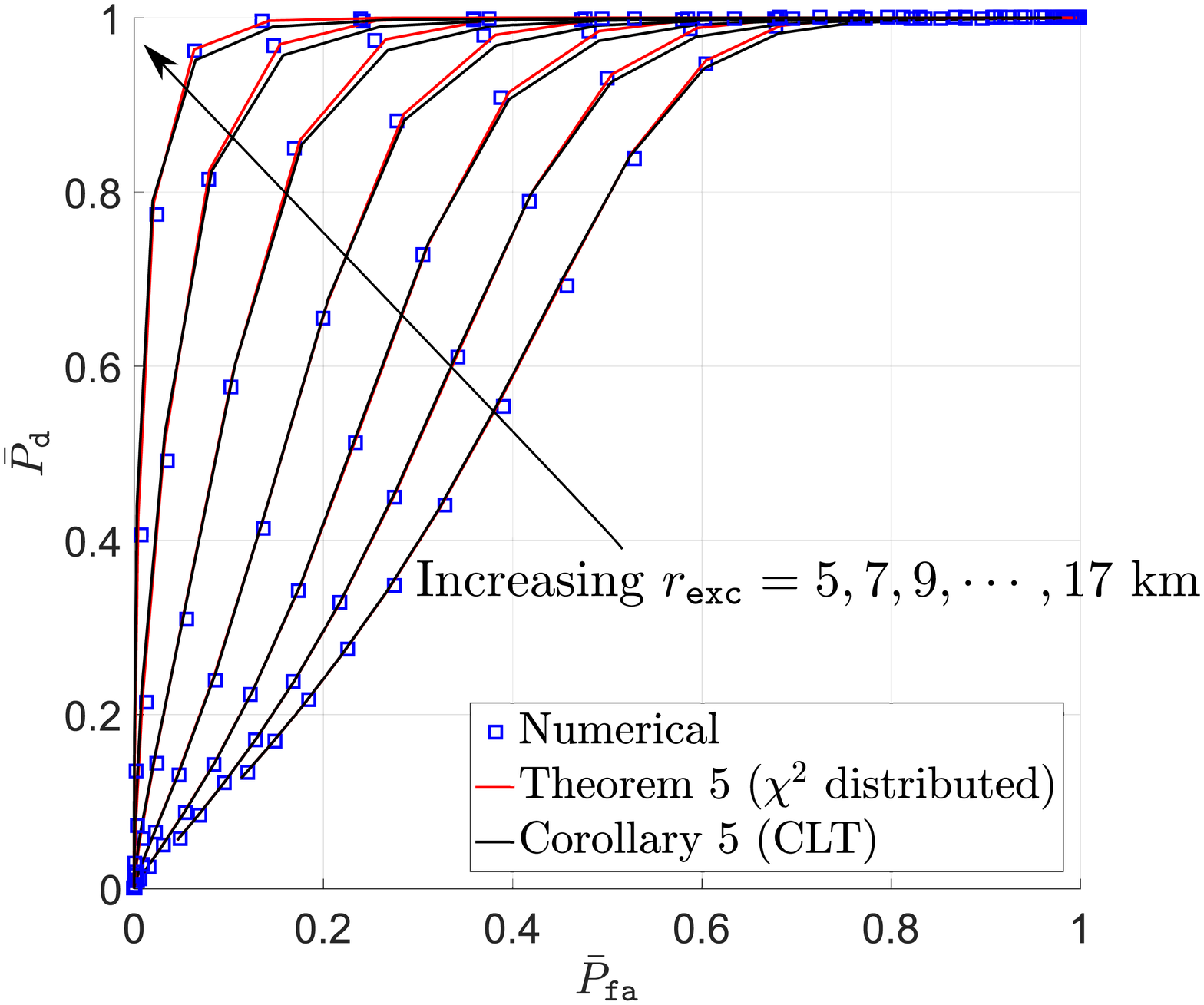}
		\caption{}
		\label{subfig:ROC_fa_d_10milli}		
	\end{subfigure}
	~
	\begin{subfigure}[t]{0.48\textwidth}
		\centering
		\includegraphics[width=2.8in]{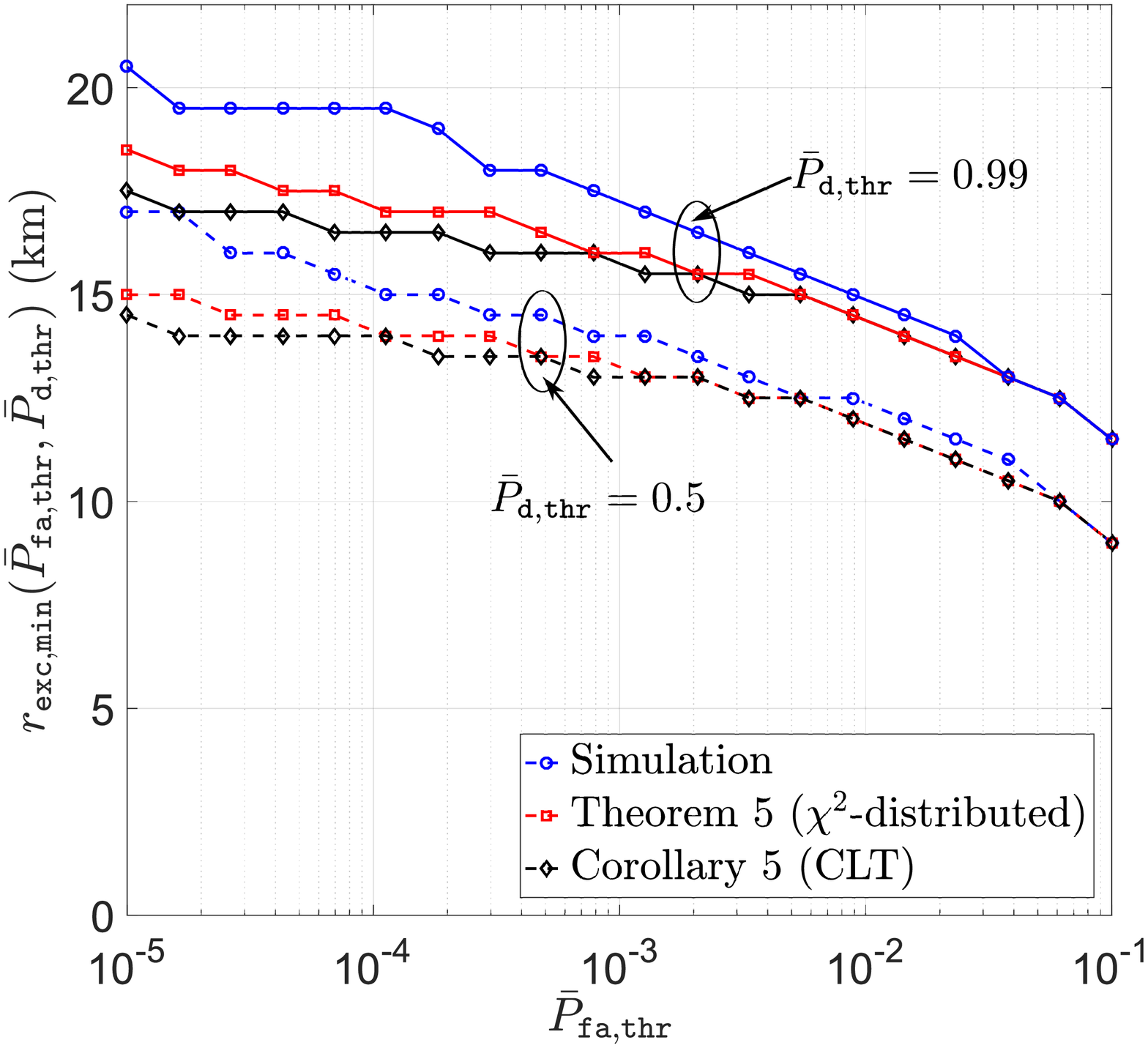}
		\caption{}
		\label{subfig:Min_r_exc_so_radar_can_operate}
	\end{subfigure}
	\caption{Variation of (a) spatial probability of detection ($\bar{P}_\mathtt{d}$), and (b) spatial probability of false alarm ($\bar{P}_\mathtt{fa}$) as a function of the detection threshold ($P_\mathtt{th}$) for different $r_\mathtt{exc}$ values. (c) ROC curve for different $r_\mathtt{exc}$ values. (d) Minimum exclusion zone radius that ensures that radar can satisfy the detection and false alarm performance requirements. $\lambda_\mathtt{BS} = 0.01 \text{ km}^{-2}$, $P_\mathtt{tar} = 10^{-7}$ W, $\sigma^2_w = 10^{-9}$ W, $h_\mathtt{rad} = 20$ m, $h_\mathtt{BS} = 50$ m, $N^{(\mathtt{BS})}_\mathtt{az}= N^{(\mathtt{BS})}_\mathtt{el} = 10, N^{(\mathtt{rad})}_\mathtt{az} = N^{(\mathtt{rad})}_\mathtt{el} = 10, \theta_\mathtt{rad} = 60^\circ$, and $\phi_\mathtt{rad} = -10^\circ, N = 10.$}
	\label{fig:Spat_Pd_Pfa_and_ROCs_10milli}
\end{figure*}
\vspace{-5pt}
\subsection{Radar Performance Metrics}
To make spectrum sharing feasible for a radar system on average, it needs to have a low spatial probability of false alarm ($\bar{P}_\mathtt{fa}$) and a high spatial probability of detection ($\bar{P}_\mathtt{d}$) in the presence of interference and noise. If the corresponding probability thresholds are $\bar{P}_\mathtt{fa,thr}$ and $\bar{P}_\mathtt{d, thr}$, then it is feasible for the radar to allow the cellular network to coexist if $\bar{P}_\mathtt{fa} \leq \bar{P}_\mathtt{fa, thr}$ and $\bar{P}_\mathtt{d} \geq \bar{P}_\mathtt{d, thr}$. Fig. \ref{subfig:Spat_d_10milli} and  Fig. \ref{subfig:Spat_fa_10milli} shows $\bar{P}_\mathtt{d}$ and $\bar{P}_\mathtt{fa}$ as a function of the detection threshold ($P_\mathtt{th}$) and $r_\mathtt{exc}$. We observe that both $\bar{P}_\mathtt{d}$ and $\bar{P}_\mathtt{fa}$ monotonically decrease with $P_\mathtt{th}$ for a fixed $r_\mathtt{exc}$. This can be explained by Corollary \ref{corollary:CLT_approx_N_infty}: since $Q(x)$ monotonically decreases with increasing $x$, increasing $P_\mathtt{th}$ reduces the integrand, thus reducing $\bar{P}_\mathtt{d}$ and $\bar{P}_\mathtt{fa}$. On the other hand, we observe the same trends when $r_\mathtt{exc}$ increases while holding the other parameters constant, which can be explained as follows. From Theorem \ref{theorem:I_tot_dom_int_method}, we observe that increasing $r_\mathtt{exc}$ reduces $I_\mathtt{tot}$. Using this insight in (\ref{equation:prob_spatial_fa_approx_CLT}), it can be seen that the term $\frac{P_\mathtt{th}}{I_\mathtt{tot} + \sigma^2_w}$ monotonically increases with increasing $r_\mathtt{exc}$, which explains the trends observed in Fig. \ref{subfig:Spat_fa_10milli}. On the other hand, by deriving the slope of $\frac{(P_\mathtt{th} - P_\mathtt{tar} - I_\mathtt{tot} - \sigma^2_w)}{\sqrt{(P_\mathtt{tar} + I_\mathtt{tot} + \sigma^2_w)^2 - P^2_\mathtt{tar}}}$ w.r.t. $I_\mathtt{tot}$ in (\ref{equation:prob_spatial_d_approx_CLT}), it can be proved that $\bar{P}_\mathtt{d}$ decreases with increasing $r_\mathtt{exc}$ when $P^2_\mathtt{tar} < P_\mathtt{th}(P_\mathtt{tar} + I_\mathtt{tot} + \sigma^2_w)$. Since $P_\mathtt{th} > P_\mathtt{tar}$ for the most part in Fig. \ref{subfig:Spat_d_10milli}, this insight is consistent with our observations\footnote{It is worthwhile to mention that the $\bar{P}_\mathtt{d}$ vs $P_\mathtt{th}$ curves corresponding to different $r_\mathtt{exc}$ can intersect and cross-over. For example, this occurs if $\{P^2_\mathtt{tar} < P_\mathtt{th}(P_\mathtt{tar} + I_\mathtt{tot} + \sigma^2_w)\}$ is true for some $P_\mathtt{th}$ values, and false for others.}. Fig. \ref{subfig:ROC_fa_d_10milli} shows the ROC curve for different $r_\mathtt{exc}$ values. We observe that the trends follow Figs. \ref{subfig:Spat_d_10milli}-\ref{subfig:Spat_fa_10milli}, and that the analytical and simulation results match. However, the inaccuracy due to the CLT approximation can be observed in the high $\bar{P}_\mathtt{d}$-low $\bar{P}_\mathtt{fa}$ regime, which is likely due to the difference in tail behavior of the Gaussian and $\chi^2$-distributions. 
\vspace{-5pt}

\subsection{Designing the Minimum Exclusion Zone Radius for Radar-Cellular Coexistence}
From Theorem \ref{theorem:P_fa_P_d_first_princip_DIA} and Corollary \ref{corollary:P_fa_P_d_first_princip_DIA}, we observe that $\bar{P}_\mathtt{d}$ and $\bar{P}_\mathtt{fa}$ are dependent on $r_\mathtt{exc}$. Conditioned on the other parameters, the minimum $r_\mathtt{exc}$ ($r_\mathtt{exc,min}$) for which the radar can coexist with the cellular system without significant performance degradation is given by 
\begin{align}
\label{eq:Min_r_exc_radar_operate}
r_\mathtt{exc, min}(\bar{P}_\mathtt{d, thr}, \bar{P}_\mathtt{fa, thr}) = & \underset{r_\mathtt{exc} \in \mathbb{R}^+}{\min} \{r_\mathtt{exc}|\bar{P}_\mathtt{d}(r_\mathtt{exc}) \geq \bar{P}_\mathtt{d, thr}, \nonumber \\
& \bar{P}_\mathtt{fa}(r_\mathtt{exc}) \leq \bar{P}_\mathtt{fa, thr} \},
\end{align}
where $\bar{P}_\mathtt{d, thr}$ and $\bar{P}_\mathtt{fa, thr}$ are the spatial probability of detection/false alarm thresholds. Unfortunately, the above optimization problem is intractable since $r_\mathtt{exc}$ lies in the lower limit of the integral. However, this can be solved using numerical methods since $\bar{P}_\mathtt{d}$ monotonically increases with $r_\mathtt{exc}$ (Fig. \ref{subfig:Spat_d_10milli}), and $\bar{P}_\mathtt{fa}$ decreases monotonically with $r_\mathtt{exc}$ (Fig. \ref{subfig:Spat_fa_10milli}). Therefore, we numerically solve (\ref{eq:Min_r_exc_radar_operate}) by restricting it to a finite set $\mathcal{R}_{\rm exc}$, i.e. 
\begin{align}
\label{eq:Min_r_exc_radar_operate_discretized}
r_\mathtt{exc, min}(\bar{P}_\mathtt{d, thr}, \bar{P}_\mathtt{fa, thr}) = & \underset{r_\mathtt{exc} \in \mathcal{R}_{\rm exc}}{\min} \{r_\mathtt{exc}|\bar{P}_\mathtt{d}(r_\mathtt{exc}) \geq \bar{P}_\mathtt{d, thr}, \nonumber \\
& \bar{P}_\mathtt{fa}(r_\mathtt{exc}) \leq \bar{P}_\mathtt{fa, thr} \}.
\end{align}
Fig. \ref{subfig:Min_r_exc_so_radar_can_operate} shows the results for the above minimization problem as a function of $\bar{P}_\mathtt{d, thr}$ and $\bar{P}_\mathtt{fa, thr}$, where $\mathcal{R}_{\rm exc} = \{0 \text{ km}, 0.5\text{ km}, 1\text{ km} \cdots, 35\text{ km} \}$. We observe that the $r_\mathtt{exc, min}$ obtained using our theoretical expressions lie within $6 \%$ of the simulation result on average. We would like to emphasize that system-level simulations of radar-cellular spectrum sharing scenarios is time consuming, especially in conditions of high $\lambda_\mathtt{BS}$ and $r_\mathtt{exc}$. Our analytical expressions are valuable in preliminary feasibility studies to quickly obtain estimates of system parameters for which spectrum sharing is feasible, to aid system-level simulations.
\vspace{-5pt}
\section{Conclusion and Proposed Work}\label{Sec_Conc_Prop_Work}
In this paper, we presented an analytical framework to evaluate radar performance metrics in underlay radar-massive MIMO cellular spectrum sharing scenarios, where both systems are equipped with 3D beamforming capabilities. We devised a novel construction based on bounding a PV cell by its circumcircle, to upper bound the worst-case average interference at the radar due to a co-channel massive MIMO downlink in near LoS channel conditions. We also proposed and analyzed the nominal average and variance of the interference power using a more tractable model, where each cell is replaced by a circular disk of area equal to the average area of a typical cell. We provided useful insights regarding the worst-case exclusion zone radius, scaling of interference power with BS density, and the approximate gap between the worst-case and nominal average interference power. We then derived the \textit{equi-interference contour} under the nominal interference model, and used it to characterize the interference distribution, using the \textit{dominant interference approximation}. Under a quasi-static target detection scenario based on coherent integration across multiple radar pulses and threshold detection, we used the interference distribution to characterize the spatial probability of detection and false alarm. 

We showed that the upper bound using the circumcircle-based model is remarkably tight for 3GPP deployment parameters \cite{3GPP5GNR_ChanModels}, and then demonstrated the usefulness of our proposed approach by applying it for evaluation of radar performance metrics, especially ROC curves. We also provided intuitive system design insights to explain the accuracy of the dominant interferer method, and the trends in the radar's detection and false alarm performance. Finally, we applied our analytical results to design the minimum exclusion zone radius to enable radar-cellular coexistence. The analytical framework presented in this paper (a) enables network designers to systematically isolate and evaluate the impact of each deployment parameter (BS density, antenna height, transmit power, exclusion zone radius etc.) on the worst-case radar performance, and (b) complements industry-standard simulation methodologies, by establishing a baseline for each set of deployment parameters in practical spectrum sharing scenarios. 

This work focused on studying the impact of worst-case cellular interference on radar performance. Hence, characterizing the impact of cellular uplink interference on radar performance and analyzing the throughput and spectral efficiency performance of the cellular downlink/uplink in the presence of radar interference is an important extension to this work. From a harmonious coexistence perspective, using this work to progress towards system-level optimization frameworks that seek to maximize the radar performance under cellular quality of service (QoS) constraints, and vice-versa, is an important research direction. 
\vspace{-12pt}
\appendix
\subsection{Proof of Lemma \ref{Lemma_Monotonic_BF_Gain}}\label{App1_Proof_BFGain_UpBound}
The steering vector of a $N_\mathtt{az} \times N_\mathtt{el}$ URA is $\mathbf{a}(\theta,\phi)=\mathbf{a}_\mathtt{az}(\theta, \phi) \otimes \mathbf{a}_\mathtt{el}(\phi)$, where $\otimes$ is the Kronecker product. For $\tfrac{\lambda}{2}$-spacing, $\mathbf{a}_\mathtt{az}(\theta,\phi) = [1\ e^{-j \pi \sin \theta \cos \phi}\cdots e^{-j \pi (N_\mathtt{az} - 1)\sin \theta \cos \phi}] \in \mathbb{C}^{N_\mathtt{az} \times 1}$, 
$\mathbf{a}_{el}(\phi) = [1\ e^{-j\pi \sin \phi}\cdots e^{-j\pi (N_\mathtt{el} - 1)\sin \phi}] \in \mathbb{C}^{N_\mathtt{el} \times 1}$.
Using the properties of the Kronecker product, expanding and simplifying, we get
\begin{align}
\label{BFGain_Geom_Series_Expansion}
G_\mathtt{BS}(\theta,\phi,\theta_k, \phi_k) = & \tfrac{\sin^2 \big(\tfrac{\pi}{2} N_\mathtt{az}(\sin\theta \cos \phi - \sin \theta_k \cos \phi_k) \big)}{N_\mathtt{az}\sin^2 \big(\tfrac{\pi}{2} (\sin\theta \cos \phi - \sin \theta_k \cos \phi_k) \big)} \cdot \nonumber \\
& \tfrac{\sin^2 \big(\tfrac{\pi}{2} N_\mathtt{el}(\sin \phi - \sin \phi_k) \big)}{N_\mathtt{el}\sin^2 \big(\tfrac{\pi}{2} (\sin \phi - \sin \phi_k) \big)} \leq N_\mathtt{az} N_\mathtt{al}.
\end{align}
Since $\tfrac{\sin^2(Na)}{\sin^2a} \leq N^2 \text{ for } a \in \mathbb{R}$, the universal upper bound is obtained above, and is achieved when $a = 0$. To obtain a tighter bound $G^{(\mathtt{max})}_\mathtt{BS}$ defined in (\ref{BFGain_Tight_UpperBound}), we consider the following cases.
\paragraph*{Case 1} If $\phi_\mathtt{m} \leq \phi \leq \tfrac{\pi}{2}$, $G_\mathtt{BS}(\theta, \phi,\theta_k, \phi_k)$ is maximized by $\phi_k = \phi$, $\theta_k = \theta$, yielding $G^{(\mathtt{max})}_\mathtt{BS}(\phi, \phi_\mathtt{m}) = N_\mathtt{az} N_\mathtt{el}$.

\paragraph*{Case 2}
By upper bounding the \textit{azimuth beamforming gain} in (\ref{BFGain_Geom_Series_Expansion}), we get $G_\mathtt{BS}(\theta,\phi,\theta_k, \phi_k) \leq N_\mathtt{az} \tfrac{\sin^2 \big(\tfrac{\pi}{2} N_\mathtt{el}(\sin \phi - \sin \phi_k) \big)}{N_\mathtt{el}\sin^2 \big(\tfrac{\pi}{2} (\sin \phi - \sin \phi_k) \big)}$. The RHS monotonically decreases w.r.t. $\phi_k$ when $0 \leq \sin \phi_\mathtt{m} \leq \tfrac{1 + N_\mathtt{el} \sin \phi}{N_\mathtt{el}} \leq \tfrac{\pi}{2}$ and hence, the upper bound is $G^{(\mathtt{max})}_\mathtt{BS}(\phi, \phi_\mathtt{m}) = \tfrac{N_\mathtt{az}  \sin^2 \big(\tfrac{\pi}{2} N_\mathtt{el}(\sin \phi - \sin \phi_\mathtt{m}) \big)}{N_\mathtt{el}\sin^2 \big(\tfrac{\pi}{2} (\sin \phi - \sin \phi_\mathtt{m}) \big)}$.

\paragraph*{Case 3}
If $\tfrac{1 + N_\mathtt{el} \sin \phi}{N_\mathtt{el}} \leq \sin \phi_\mathtt{m}$, the numerator of $G^{(\mathtt{max})}_\mathtt{BS}(\cdot)$ in case 2 can be upper bounded as $\sin^2(b) \leq 1\  \forall \ b \in \mathbb{R}$, resulting in a monotonically decreasing function of $\phi_\mathtt{m}$. Hence, $G^{(\mathtt{max})}_\mathtt{BS}(\phi, \phi_\mathtt{m}) = \tfrac{N_\mathtt{az}}{N_\mathtt{el}\sin^2 \big(\tfrac{\pi}{2} (\sin \phi - \sin \phi_\mathtt{m}) \big)}$.
\begin{comment} 
{\color{blue}We note that the upper bound on the beamforming gain is independent of the azimuth angle, since the maximum azimuth beamforming gain is universally upper bounded by $N_\mathtt{az}$, a constant. Therefore for ease of exposition, we have assumed that the boresight of each BS is aligned along the direction of the radar, which corresponds to $\theta= 0^\circ$ as discussed in Assumption \ref{BoresightAssumption}.}
\end{comment}
\vspace{-15pt}
\subsection{Proof of Theorem \ref{Circum_rad_Model}}\label{App2_Proof_Worst_Case_Int}
Since the massive MIMO BS locations are modeled as an independent PPPs $\mathbf{\Phi_{BS}}$ with intensity $\lambda_\mathtt{BS}$, the worst-case average interference at the radar is given by Campbell's theorem using
\begin{align*}
\bar{I}_\mathtt{rad, c} & = \mathbb{E} \Big[ \mathbb{E} \Big[\sum_{\mathbf{X} \in \mathbf{\Phi_{int}}} \{ I^{(\mathtt{w})}_\mathtt{rad} (\mathbf{X}, h_\mathtt{BS}, h_\mathtt{rad})|r_c \}\Big]\Big| r_c \Big] \nonumber \\
& = \mathbb{E} \Big[ \int_{\mathbf{x} \in \mathbf{\Phi_{int}} } \lambda_\mathtt{BS}\{ I^{(\mathtt{w})}_\mathtt{rad} (\mathbf{x}, h_\mathtt{BS}, h_\mathtt{rad})|r_c \} {\rm d} \mathbf{x} \Big| r_c \Big],
\end{align*}
where $\mathbf{x} = [r\cos \theta_{r,L}\ r\sin \theta_{r,L}]$, $\mathbf{\Phi_{int}} = \mathbf{\Phi_{BS}} \setminus  \{(x,y)|(x^2 + y^2) \leq r^2_\mathtt{exc}\}$, and $r_c$ is the cell radius that determines $G^{(\mathtt{max})}_\mathtt{BS}(\phi, \phi_\mathtt{m})$ in equation (\ref{BFGain_Tight_UpperBound}). 
Substituting (\ref{WorstCaseAvgIntPow_SingleBS}) above, noting that  
$\phi_{r,L}(r)=-\phi_{t,L}(r)=\tan^{-1} \big( \tfrac{h_\mathtt{rad} - h_\mathtt{BS}}{r} \big)$, and converting to polar coordinates we get 
\begin{align*}
\bar{I}_\mathtt{rad,c} = & \mathbb{E} \Big[ \int_{r_\mathtt{exc}}^{\infty} \int_{-\tfrac{\pi}{2}}^{\tfrac{\pi}{2}} \lambda_\mathtt{BS} \beta(d) G_\mathtt{rad}(\theta_\mathtt{rad}, \phi_\mathtt{rad}, \theta_{r,L}, \phi_{r,L}(r)) \cdot \nonumber \\
& G^{(\mathtt{max})}_\mathtt{BS}(\phi_{t,L}(r), \phi_\mathtt{m}(r_c)) \tfrac{ P_{BS}}{K}  r {\rm d}r {\rm d}\theta_{r,L} \Big| r_c \Big],
\end{align*}
where $d = \sqrt{r^2 + (h_\mathtt{BS} - h_\mathtt{rad})^2}$, and $\beta(d) = PL(r_0) d^{-\alpha}$ is the pathloss model. Using these and integrating over $r_c \sim f_{R_c} (r_c)$, we get the desired result. 
\vspace{-10pt}
\subsection{Proof of Lemma \ref{lemma:Contour_large_exc_zone}}\label{appendix:proof_of_contour}
Since $r_\mathtt{exc}$ is much larger than the antenna heights, we have $\phi(r)\rightarrow 0$ for $r \geq r_\mathtt{exc}$ in equation  (\ref{eq:contour_line_general_exp}) and (\ref{eq:contour_line_general_exp_proof}). Using this, the radar beamforming gain can be upper bounded similar to (\ref{BFGain_Geom_Series_Expansion}) as
\begin{align}
	G_\mathtt{rad}(\theta_\mathtt{rad},\phi_\mathtt{rad},\theta,0) = & \tfrac{\sin^2 \big(\tfrac{\pi}{2} N^{(\mathtt{rad})}_\mathtt{az}(\sin\theta_\mathtt{rad} \cos \phi_\mathtt{rad} - \sin \theta) \big)}{N^{(\mathtt{rad})}_\mathtt{az} \sin^2 \big(\tfrac{\pi}{2} (\sin\theta_\mathtt{rad} \cos \phi_\mathtt{rad} - \sin \theta) \big) } \cdot \nonumber \\
	& \tfrac{\sin^2 \big(\tfrac{\pi}{2} N^{(\mathtt{rad})}_\mathtt{el}\sin \phi_\mathtt{rad} \big)}{N^{(\mathtt{rad})}_\mathtt{el} \sin^2 \big(\tfrac{\pi}{2} \sin \phi_\mathtt{rad} \big)} \nonumber \\
	\leq & \tfrac{N^{(\mathtt{rad})}_\mathtt{az} \sin^2 \big(\tfrac{\pi}{2} N^{(\mathtt{rad})}_\mathtt{el}\sin \phi_\mathtt{rad} \big)}{N^{(\mathtt{rad})}_\mathtt{el}\sin^2 \big(\tfrac{\pi}{2} \sin \phi_\mathtt{rad} \big)}.
\end{align}
We note that the maximum azimuth beamforming gain of $N_\mathtt{az}$ is always achieved at $\theta_\mathtt{max} = \sin^{-1} (\sin \theta_\mathtt{rad} \cos \phi_\mathtt{rad})$. Therefore, the maximum radar beamforming gain is only a function of $\phi_\mathtt{rad}$. For similar reasons, when $\phi(r) \rightarrow 0$, $G^{(\mathtt{max})}_{BS}(\cdot)$ is only a function of the minimum elevation angle, which in turn is a function of $h_\mathtt{BS}\sqrt{\lambda_\mathtt{BS}}$. Defining $I_\mathtt{dom}$ to be the interference power due to the BS at $(r_\mathtt{dom},\theta_\mathtt{max})$, given by $I_\mathtt{dom} = \tfrac{P_{BS}PL(r_0)G^{(\mathtt{max})}_\mathtt{BS}(0, \phi_\mathtt{m}(1/\sqrt{\pi \lambda_\mathtt{BS}}))}{Kr^\alpha_\mathtt{dom}}\cdot \tfrac{N^{(\mathtt{rad})}_\mathtt{az} \sin^2 \big( \tfrac{\pi}{2} N^{(\mathtt{rad})}_\mathtt{el} \sin \phi_\mathtt{rad} \big)}{N^{(\mathtt{rad})}_\mathtt{el} \sin^2 \big( \tfrac{\pi}{2} \sin \phi_\mathtt{rad}  \big)}, r_\mathtt{dom} \geq r_\mathtt{exc}$. Substituting this into  (\ref{eq:contour_line_general_exp_proof}) and simplifying, we get the analytical expression of $\mathcal{C}(I_\mathtt{dom})$.
\vspace{-10pt}
\subsection{Proof of Lemma \ref{lemma:area_B_r_closedform}} \label{appendix:proof_of_B_r_area}
Let $\mathcal{A}(r_\mathtt{dom})$ denote the region outside the exclusion zone enclosed by $\mathcal{C}(r_\mathtt{dom})$, and $A(r_\mathtt{dom})$ denote the corresponding area. Using equation (\ref{eq:contour_line_general_exp}), this region can be written as
\begin{align*}
	\mathcal{A}(r_\mathtt{dom}) = & \Big \{  (r, \theta) \Big| \tfrac{-\pi}{2} \leq \theta \leq \tfrac{\pi}{2}, r_\mathtt{exc} \leq r \leq \max\Big( r_\mathtt{exc}, r_\mathtt{dom} \times \nonumber \\
	& \Big[ \tfrac{\sin \big(\tfrac{\pi}{2} N^{(\mathtt{rad})}_\mathtt{az} (\sin \theta_\mathtt{rad} \cos \phi_\mathtt{rad} - \sin \theta) \big)}{N^{(\mathtt{rad})}_\mathtt{az} \sin^2 \big( \tfrac{\pi}{2} (\sin \theta_\mathtt{rad} \cos \phi_\mathtt{rad} - \sin \theta) \big)} \Big]^{\frac{2}{\alpha}} \Big)\Big \}.
\end{align*}
Defining $\tilde{r}_\mathtt{dom}(\theta) \triangleq r_\mathtt{dom} \Big[ \tfrac{\sin \big(\tfrac{\pi}{2} N^{(\mathtt{rad})}_\mathtt{az} (\sin \theta_\mathtt{rad} \cos \phi_\mathtt{rad} - \sin \theta) \big)}{N^{(\mathtt{rad})}_\mathtt{az} \sin \big( \tfrac{\pi}{2} (\sin \theta_\mathtt{rad} \cos \phi_\mathtt{rad} - \sin \theta) \big)} \Big]^{\frac{2}{\alpha}}$ and using the above, the area $A(r_\mathtt{dom})$ is given by
\begin{align}
	\label{eq:derive_B_r_closedform2}
	A(r_\mathtt{dom}) &  = \int_{-\frac{\pi}{2}}^{\frac{\pi}{2}} \int_{r_\mathtt{exc}}^{\max(r_\mathtt{exc},\tilde{r}_\mathtt{dom}(\theta))} r {\rm d} r {\rm d} \theta  \nonumber \\
	& = \frac{1}{2} \int_{-\frac{\pi}{2}}^{\frac{\pi}{2}} \max \big(r^2_\mathtt{exc}, \tilde{r}^2_\mathtt{dom}(\theta) \big) {\rm d} \theta - \frac{\pi r^2_\mathtt{exc}}{2}.
\end{align}
Expanding and simplifying, we get the desired result.
\vspace{-5pt}
\subsection{Proof of Lemma \ref{lemma:r_dom_density}}\label{appendix:lemma_r_dom_density_proof}
The distribution of $r_\mathtt{dom}$ is given by $F_{R_\mathtt{dom}}(r_\mathtt{dom}) = \mathbb{P}[R_\mathtt{dom} \leq r_\mathtt{dom}]$. Since the area outside the exclusion zone enclosed by the contour is $A(r_\mathtt{dom})$, the CDF is the void probability given by
\begin{align}
	\label{eq:void_prob_PPP_B_r_dom_new}
	F_{R_\mathtt{dom}}(r_\mathtt{dom}) = 1 - \exp \big( -\lambda_\mathtt{BS} A(r_\mathtt{dom})\big), \text{ for } r_\mathtt{dom} \geq r_\mathtt{exc}. 
\end{align}
Substitution equation (\ref{eq:B_r_dom_closedform}) in the above, we get the desired CDF. Further, differentiating equation (\ref{eq:void_prob_PPP_B_r_dom_new}), the density of $r_\mathtt{dom}$ can be written as $f_{R_\mathtt{dom}}(r_\mathtt{dom}) = \frac{{\rm d} A(r_\mathtt{dom})}{{\rm d} r_\mathtt{dom}} \cdot \lambda_\mathtt{BS} e^{-\lambda_\mathtt{BS} A(r_\mathtt{dom})}, \text{ for } r_\mathtt{dom} \geq r_\mathtt{exc}$. Due to the presence of the $\max(\cdot)$ term in equation (\ref{eq:B_r_dom_closedform}), it can be shown that $A(r_\mathtt{dom})$ depends on $r_\mathtt{dom}$ only in certain ranges of $\theta$, which can also be observed in Fig. \ref{subfig:Illustration_of_contour_area_as_fn_of_r1}. Hence, we get
\begin{align}
	\label{eq:derivative_B_r_dom_cases}
	& \frac{{\rm d} [\max(r^2_\mathtt{exc}, \tilde{r}^2_\mathtt{dom}(\theta))]}{{\rm d} r_\mathtt{dom}} = \nonumber \\
	& \begin{cases}
		0  \qquad \qquad  \qquad \qquad \qquad \qquad \qquad \quad \text{ if } \tilde{r}_\mathtt{dom} (\theta) < r_\mathtt{exc} \\
		2r_\mathtt{dom} \Big[ \tfrac{\sin^2 \big(\tfrac{\pi}{2} N_\mathtt{az} (\sin \theta_\mathtt{rad} \cos \phi_\mathtt{rad} - \sin \theta) \big)}{N^2_\mathtt{az} \sin^2 \big( \tfrac{\pi}{2} (\sin \theta_\mathtt{rad} \cos \phi_\mathtt{rad} - \sin \theta) \big)} \Big]^{2/\alpha}  \qquad  \text{ otherwise,} \\		
	\end{cases}
\end{align}
Substituting this into $f_{R_\mathtt{dom}}(r_\mathtt{dom})$ and representing it in terms of $\mathbbm{1} [\cdot]$, we obtain the desired result.
\vspace{-5pt}
\subsection{Proof of Theorem \ref{theorem:I_tot_dom_int_method}} \label{appendix:proof_of_I_tot_dom_int_apprx}
The dominant interference power is given by $I_\mathtt{dom} = \kappa r^{-\alpha}_\mathtt{dom}$. Next, we compute the average interference power due to the rest of the network, conditioned on $I_\mathtt{dom}$, i.e. $\mathbb{E}[I_\mathtt{rest}|I_\mathtt{dom}]$. Due to the bijection between $r_\mathtt{dom}$ and $I_\mathtt{dom}$ in the AAECC model, we have $\mathbb{E}[I_\mathtt{rest}|I_\mathtt{dom}] = \mathbb{E}[I_\mathtt{rest}|r_\mathtt{dom}]$. Hence, we can compute the conditional average interference power using
\begin{align}
	\label{equation:Avg_est_condn_on_I_dom}
	& \mathbb{E}[I_\mathtt{rest}|I_\mathtt{dom}] = \tfrac{P_{BS} \lambda_\mathtt{BS} G_\mathtt{BS}(0, \phi_\mathtt{m}(1/\sqrt{\pi \lambda_\mathtt{BS}})) PL(r_0) }{K} \times \nonumber \\
	&\qquad  \int_{-\tfrac{\pi}{2}}^{\tfrac{\pi}{2}} \int_{\max (r_\mathtt{exc}, \tilde{r}_\mathtt{dom}(\theta) )}^{\infty}  G_\mathtt{rad}(0, \phi_\mathtt{rad}, \theta, 0) r^{-\alpha + 1} {\rm d}r {\rm d}\theta \nonumber \\
	& \stackrel{(a)}{=} \frac{\kappa}{\alpha - 2} \int_{-\tfrac{\pi}{2}}^{\tfrac{\pi}{2}}  \big[\max\big(r_\mathtt{exc}, \tilde{r}_\mathtt{dom}(\theta) \big) \big]^{-\alpha + 2} \tfrac{G_\mathtt{rad}(\theta_\mathtt{rad}, \phi_\mathtt{rad}, \theta, 0)}{G_\mathtt{rad}(0, \phi_\mathtt{rad}, 0, 0)} {\rm d}\theta.
\end{align} 
The equality in (a) is obtained by  defining $\kappa \triangleq \frac{P_{BS} \lambda_\mathtt{BS} G_\mathtt{rad}(0,\phi_\mathtt{rad},0,0) G_\mathtt{BS}(0, \phi_\mathtt{m}(1/\sqrt{\pi \lambda_\mathtt{BS}})) PL(r_0) }{K} $, and evaluating the inner integral. Using Lemma \ref{lemma:Contour_large_exc_zone} and equation (\ref{equation:Avg_est_condn_on_I_dom}) in equation (\ref{equation:dom_int_method_defn}) and simplifying, we get the desired result. 
\vspace{-5pt}
\subsection{Proof of Theorem \ref{theorem:P_fa_P_d_first_princip_DIA}} \label{appendix:proof_of_spat_Pd_Pfa_spat_chisq}
We note that under hypothesis $\mathcal{H}_i$, the received power is $P_{\mathtt{rad},i}$ for $i=\{ 0,1 \}$. By definition we have $\mathbb{P}[P_{\mathtt{rad},i} > P_\mathtt{th}|I_\mathtt{tot}] = 1 - F_{P_\mathtt{rad},i}(P_\mathtt{th}|I_\mathtt{tot})$. Therefore, using Lemma \ref{lemma:central_non_central_chi_sq_dist} in (\ref{equation:prob_spatial_det_fa_defn}), we get
\begin{align}
	\label{eq:proof_of_spat_Pd_Pfa_chi_sq_1}
	\bar{P}_\mathtt{d} & = 1 - \int_{0}^\infty Q_N \Big( \sqrt{\tfrac{2 N P_\mathtt{tar}}{I_\mathtt{tot} + \sigma^2_w}}, \sqrt{\tfrac{2 N P_\mathtt{th}}{I_\mathtt{tot} + \sigma^2_w}} \Big) f_{I_{\mathtt{tot}}}(x) {\rm d}x, \nonumber \\
	\bar{P}_\mathtt{fa} & = 1 - \int_{0}^\infty\tfrac{1}{ (N-1)! } \gamma_l \Big( N, \tfrac{NP_\mathtt{th}}{I_\mathtt{tot} + \sigma^2_w} \Big) f_{I_{\mathtt{tot}}}(x) {\rm d}x.	
\end{align}
The first approximation is obtained by replacing $I_\mathtt{tot}$ by $I_\mathtt{tot,DI}$ using Theorem \ref{theorem:I_tot_dom_int_method}, and changing the upper limit to $I_\mathtt{exc}+\bar{I}_\mathtt{rad,a}$ (Remark \ref{remark:finite_supp_I_tot_DI}). Using the bijection between $r_\mathtt{dom}$ and $I_\mathtt{tot,DI}$ (Corollary \ref{corollary:monoton_DI_I_tot}), the final result is obtained by substituting $I_\mathtt{tot,DI}$ by $r_\mathtt{dom}$, and applying the chain rule. 

\bibliographystyle{IEEEtran}
\bibliography{IEEE_TWC_references}

\begin{thebibliography}{10}
\providecommand{\url}[1]{#1}
\csname url@samestyle\endcsname
\providecommand{\newblock}{\relax}
\providecommand{\bibinfo}[2]{#2}
\providecommand{\BIBentrySTDinterwordspacing}{\spaceskip=0pt\relax}
\providecommand{\BIBentryALTinterwordstretchfactor}{4}
\providecommand{\BIBentryALTinterwordspacing}{\spaceskip=\fontdimen2\font plus
\BIBentryALTinterwordstretchfactor\fontdimen3\font minus
  \fontdimen4\font\relax}
\providecommand{\BIBforeignlanguage}[2]{{%
\expandafter\ifx\csname l@#1\endcsname\relax
\typeout{** WARNING: IEEEtran.bst: No hyphenation pattern has been}%
\typeout{** loaded for the language `#1'. Using the pattern for}%
\typeout{** the default language instead.}%
\else
\language=\csname l@#1\endcsname
\fi
#2}}
\providecommand{\BIBdecl}{\relax}
\BIBdecl

\bibitem{Rao_Dhillon_Globecom_2019}
R.~M. Rao, H.~S. Dhillon, V.~Marojevic, and J.~H. Reed, ``{Analysis of
  Worst-Case Interference in Underlay Radar-Massive MIMO Spectrum Sharing
  Scenarios},'' in \emph{in Proc. IEEE GLOBECOM}, 2019, pp. 1--6.

\bibitem{FCC_AWS3_auction}
\BIBentryALTinterwordspacing
``{Auction 97: Advanced Wireless Services (AWS-3) Fact Sheet},'' \emph{{Federal
  Communications Commission}}, Jan 2015. [Online]. Available:
  \url{https://www.fcc.gov/auction/97/factsheet}
\BIBentrySTDinterwordspacing

\bibitem{FCC_3point5_GHz_Rules}
{FCC}, ``{Amendment of the Commission's Rules with Regard to Commercial
  Operations in the 3550-3650 MHz Band},'' \emph{{Federal Communications
  Commission, Report and Order and Second Further Notice of Proposed
  Rulemaking}}, April 2015.

\bibitem{FCC_5_GHz_FirstOrder}
------, ``{Revision of Part 15 of the Commission's Rules to Permit Unlicensed
  National Information Infrastructure (U-NII) Devices in the 5 GHz Band},''
  \emph{{Federal Communications Commission, First Report and Order}}, April
  2014.

\bibitem{Griffiths_Radar_tech_regul_PRoc_IEEE_2015}
H.~{Griffiths} \emph{et~al.}, ``{Radar Spectrum Engineering and Management:
  Technical and Regulatory Issues},'' \emph{Proc. IEEE}, vol. 103, no.~1, pp.
  85--102, 2015.

\bibitem{Khan_DaSilva_WCM_2016}
Z.~Khan, J.~J. Lehtomaki, R.~Vuohtoniemi, E.~Hossain, and L.~A. Dasilva, ``{On
  Opportunistic Spectrum Access in Radar Bands: Lessons Learned from
  Measurement of Weather Radar Signals},'' \emph{IEEE Wireless Communications},
  vol.~23, no.~3, pp. 40--48, June 2016.

\bibitem{Hessar_Roy_Radar_WiFi_TAES_2016}
F.~Hessar and S.~Roy, ``{Spectrum Sharing between a Surveillance Radar and
  Secondary Wi-Fi Networks},'' \emph{IEEE Transactions on Aerospace and
  Electronic Systems}, vol.~52, no.~3, pp. 1434--1448, June 2016.

\bibitem{Rao_Vuk_TVT_2020}
R.~M. Rao, V.~Marojevic, and J.~H. Reed, ``{Semi-Blind Post-Equalizer SINR
  Estimation and Dual CSI Feedback for Radar-Cellular Coexistence},''
  \emph{IEEE Transactions on Vehicular Technology}, vol.~69, no.~9, pp.
  9720--9735, 2020.

\bibitem{sanders2006effects}
F.~H. Sanders \emph{et~al.}, \emph{{Effects of RF Interference on Radar
  Receivers}}.\hskip 1em plus 0.5em minus 0.4em\relax US Department of
  Commerce, NTIA, 2006.

\bibitem{Li_Baccelli_Andrews_LTE_WiFi_TWC_2016}
Y.~Li, F.~Baccelli, J.~G. Andrews, T.~D. Novlan, and J.~C. Zhang, ``{Modeling
  and Analyzing the Coexistence of Wi-Fi and LTE in Unlicensed Spectrum},''
  \emph{IEEE Transactions on Wireless Communications}, vol.~15, no.~9, pp.
  6310--6326, Sep. 2016.

\bibitem{Bodong_SSS_HetNet_TWC_2020}
B.~Shang, L.~Liu, H.~Chen, C.~J. Zhang, S.~Pudlewski, E.~S. Bentley, and J.~D.
  Ashdown, ``{Spatial Spectrum Sensing in Uplink Two-Tier User-Centric Deployed
  HetNets},'' \emph{IEEE Transactions on Wireless Communications}, vol.~19,
  no.~12, pp. 7957--7972, 2020.

\bibitem{Chen_Dhillon_Liu_QoS_D2D_StochGeom_TCOM_2019}
H.~Chen, L.~Liu, H.~S. Dhillon, and Y.~Yi, ``{QoS-Aware D2D Cellular Networks
  With Spatial Spectrum Sensing: A Stochastic Geometry View},'' \emph{IEEE
  Transactions on Communications}, vol.~67, no.~5, pp. 3651--3664, 2019.

\bibitem{Bodong_Liu_UAV_D2D_TCOM_2020}
B.~{Shang}, L.~Liu, R.~M. Rao, V.~Marojevic, and J.~H. Reed, ``{3D Spectrum
  Sharing for Hybrid D2D and UAV Networks},'' \emph{IEEE Transactions on
  Communications}, vol.~68, no.~9, pp. 5375--5389, 2020.

\bibitem{CBRS_Alliance}
{{CBRS Alliance}}, ``{Announcing OnGo Together},''
  \url{https://www.cbrsalliance.org/}, June 2016.

\bibitem{SNS_Telecom_CBRS_Projection}
\BIBentryALTinterwordspacing
{{SNS Telecom and IT}}, ``{LTE and 5G NR-Based CBRS Networks: 2020 – 2030 –
  Opportunities, Challenges, Strategies and Forecasts},'' Dec 2020. [Online].
  Available: \url{https://www.snstelecom.com/cbrs}
\BIBentrySTDinterwordspacing

\bibitem{US_DoD_Spec_Share_budget}
\BIBentryALTinterwordspacing
{{US Department of Defense Press Release}}, ``{DOD Announces USD 600 Million
  for 5G Experimentation and Testing at Five Installations},'' Oct 2020.
  [Online]. Available:
  \url{https://www.defense.gov/Newsroom/Releases/Release/Article/2376743/dod-announces-600-million-for-5g-experimentation-and-testing-at-five-installati/}
\BIBentrySTDinterwordspacing

\bibitem{Liu_Robust_MIMO_BF_Rad_Cell_Coexist_2017}
F.~Liu, C.~Masouros, A.~Li, and T.~Ratnarajah, ``{Robust MIMO Beamforming for
  Cellular and Radar Coexistence},'' \emph{IEEE Wireless Communications
  Letters}, vol.~6, no.~3, pp. 374--377, June 2017.

\bibitem{Biswas_FDMIMO_radar_coexist_TWC_2018}
S.~Biswas, K.~Singh, O.~Taghizadeh, and T.~Ratnarajah, ``{Coexistence of MIMO
  Radar and FD MIMO Cellular Systems With QoS Considerations},'' \emph{IEEE
  Transactions on Wireless Communications}, vol.~17, no.~11, pp. 7281--7294,
  Nov 2018.

\bibitem{Bica_Mitra_ICASSP_MI_Rad_Wfrm_2016}
M.~{Bica} \emph{et~al.}, ``{Mutual Information based Radar Waveform Design for
  Joint Radar and Cellular Communication Systems},'' in \emph{Proc. IEEE
  ICASSP}, March 2016, pp. 3671--3675.

\bibitem{Tang_Li_Spectr_constr_Rad_Wfm_TSP_2019}
B.~{Tang} and J.~{Li}, ``{Spectrally Constrained MIMO Radar Waveform Design
  Based on Mutual Information},'' \emph{IEEE Transactions on Signal
  Processing}, vol.~67, no.~3, pp. 821--834, Feb 2019.

\bibitem{Carrick_Reed_FRESH_TAES_2019}
M.~Carrick, J.~H. Reed, and C.~M. Spooner, ``{Mitigating
  Linear-Frequency-Modulated Pulsed Radar Interference to OFDM},'' \emph{IEEE
  Transactions on Aerospace and Electronic Systems}, vol.~55, no.~3, pp.
  1146--1159, June 2019.

\bibitem{Liu_Geraci_ICSI_HowMany_TWC_2019}
F.~Liu, A.~Garcia-Rodriguez, C.~Masouros, and G.~Geraci, ``{Interfering Channel
  Estimation in Radar-Cellular Coexistence: How Much Information Do We Need?}''
  \emph{IEEE Transactions on Wireless Communications}, vol.~18, no.~9, pp.
  4238--4253, Sep. 2019.

\bibitem{Xu_FDMIMO_Samsung_JSAC_2017}
G.~Xu, Y.~Li, J.~Yuan, R.~Monroe, S.~Rajagopal, S.~Ramakrishna, Y.~H. Nam,
  J.-Y. Seol, J.~Kim, M.~M.~U. Gul, A.~Aziz, and J.~Zhang, ``{Full Dimension
  MIMO (FD-MIMO): Demonstrating Commercial Feasibility},'' \emph{IEEE Journal
  of Selected Areas in Communications}, vol.~35, no.~8, pp. 1876--1886, Aug
  2017.

\bibitem{Sudeep_Abid_DynExcZone_DySPAN_2014}
S.~Bhattarai, A.~Ullah, J.-M.~J. Park, J.~H. Reed, D.~Gurney, and B.~Gao,
  ``{Defining Incumbent Protection Zones on the Fly: Dynamic Boundaries for
  Spectrum Sharing},'' in \emph{in Proc. IEEE DySPAN}, Sep. 2015, pp. 251--262.

\bibitem{Parida_Dhillon_CBRS_Access_2017}
P.~{Parida}, H.~S. {Dhillon}, and P.~{Nuggehalli}, ``{Stochastic Geometry-Based
  Modeling and Analysis of Citizens Broadband Radio Service System},''
  \emph{IEEE Access}, vol.~5, pp. 7326--7349, 2017.

\bibitem{Melvin_STAP_TAESM_2004}
W.~L. {Melvin}, ``{A STAP Overview},'' \emph{IEEE Aerospace and Electronic
  Systems Magazine}, vol.~19, no.~1, pp. 19--35, 2004.

\bibitem{LTEAPro_FDMIMO_Samsung_ComMag_2017}
H.~Ji, Y.~Kim, J.~Lee, E.~Onggosanusi, Y.~Nam, J.~Zhang, B.~Lee, and B.~Shim,
  ``{Overview of Full-Dimension MIMO in LTE-Advanced Pro},'' \emph{IEEE
  Communications Magazine}, vol.~55, no.~2, pp. 176--184, February 2017.

\bibitem{Bai_Heath_mmWave_TWC_2015}
T.~{Bai} and R.~W. {Heath}, ``{Coverage and Rate Analysis for Millimeter-Wave
  Cellular Networks},'' \emph{IEEE Transactions on Wireless Communications},
  vol.~14, no.~2, pp. 1100--1114, 2015.

\bibitem{Kim_WiFi_Radar_WCL_2017}
S.~{Kim} and C.~Dietrich, ``{Coexistence of Outdoor Wi-Fi and Radar at 3.5
  GHz},'' \emph{IEEE Wireless Communications Letters}, vol.~6, no.~4, pp.
  522--525, Aug 2017.

\bibitem{Rebato_Park_Zorzi_AntPattrn_mmWave_TCOM_2019}
M.~Rebato, J.~Park, P.~Popovski, E.~De~Carvalho, and M.~Zorzi, ``{Stochastic
  Geometric Coverage Analysis in mmWave Cellular Networks With Realistic
  Channel and Antenna Radiation Models},'' \emph{IEEE Transactions on
  Communications}, vol.~67, no.~5, pp. 3736--3752, 2019.

\bibitem{Kim_Visotsky_5G_mmWave_Coexist_JSAC_2017}
S.~Kim, E.~Visotsky, P.~Moorut, K.~Bechta, A.~Ghosh, and C.~Dietrich,
  ``{Coexistence of 5G With the Incumbents in the 28 and 70 GHz Bands},''
  \emph{IEEE Journal of Selected Areas in Communications}, vol.~35, no.~6, pp.
  1254--1268, June 2017.

\bibitem{Yang_BS_Downtilt_DL_Cellular_TWC_2019}
J.~Yang, M.~Ding, G.~Mao, Z.~Lin, D.-G. Zhang, and T.~H. Luan, ``{Optimal Base
  Station Antenna Downtilt in Downlink Cellular Networks},'' \emph{IEEE
  Transactions on Wireless Communications}, vol.~18, no.~3, pp. 1779--1791,
  2019.

\bibitem{JSDM_Adhikary_Caire_TIT_2013}
A.~{Adhikary}, J.~{Nam}, J.~{Ahn}, and G.~{Caire}, ``{Joint Spatial Division
  and Multiplexing-The Large-Scale Array Regime},'' \emph{IEEE Transactions on
  Information Theory}, vol.~59, no.~10, pp. 6441--6463, Oct 2013.

\bibitem{Vishnu_Dhillon_UAVs_TCOM_2017}
V.~V. {Chetlur} and H.~S. {Dhillon}, ``{Downlink Coverage Analysis for a Finite
  3-D Wireless Network of Unmanned Aerial Vehicles},'' \emph{IEEE Transactions
  on Communications}, vol.~65, no.~10, pp. 4543--4558, 2017.

\bibitem{Heath_Kontouris_Bai_Dom_Int_TSP_2013}
R.~W. Heath, M.~Kountouris, and T.~Bai, ``{Modeling Heterogeneous Network
  Interference Using Poisson Point Processes},'' \emph{IEEE Transactions on
  Signal Processing}, vol.~61, no.~16, pp. 4114--4126, 2013.

\bibitem{3GPP5GNR_ChanModels}
{ETSI}, ``{5G;Study on channel model for frequencies from 0.5 to 100 GHz (3GPP
  TR 38.901 version 14.0.0 Release 14)},'' \emph{{3GPP}}, May 2017.

\bibitem{andrews2016primer}
J.~G. Andrews, A.~K. Gupta, and H.~S. Dhillon, ``{A Primer on Cellular Network
  Analysis using Stochastic Geometry},'' \emph{arXiv preprint
  arXiv:1604.03183}, 2016.

\bibitem{Schloemann_Dhillon_Localze_TWC_2016}
J.~{Schloemann}, H.~S. Dhillon, and R.~M. Buehrer, ``{Toward a Tractable
  Analysis of Localization Fundamentals in Cellular Networks},'' \emph{IEEE
  Transactions on Wireless Communications}, vol.~15, no.~3, pp. 1768--1782,
  2016.

\bibitem{Haenggi_dists_Rand_wir_ntwks_TIT_2005}
M.~{Haenggi}, ``{On Distances in Uniformly Random Networks},'' \emph{IEEE
  Transactions on Information Theory}, vol.~51, no.~10, pp. 3584--3586, 2005.

\bibitem{mankar2020distance}
P.~D. Mankar, P.~Parida, H.~S. Dhillon, and M.~Haenggi, ``{Distance from the
  Nucleus to a Uniformly Random Point in the 0-Cell and the Typical Cell of the
  Poisson--Voronoi Tessellation},'' \emph{Journal of Statistical Physics}, vol.
  181, no.~5, pp. 1678--1698, 2020.

\bibitem{Lin_Div_measures_Shannon_TIT_1991}
J.~{Lin}, ``{Divergence Measures based on the Shannon Entropy},'' \emph{IEEE
  Transactions on Information Theory}, vol.~37, no.~1, pp. 145--151, 1991.

\end{thebibliography}
\ifCLASSOPTIONcaptionsoff
\newpage
\fi
\end{document}